%% file: main_alg.tex
\documentclass[reqno, 11pt]{amsart}

\usepackage{amsaddr, color}
\usepackage{a4wide}

\input{makros}

\begin{document}    
         
\title[Statistical Multiresolution Dantzig Estimation in Imaging]{Statistical
Multiresolution Dantzig Estimation in Imaging: Fundamental Concepts and
Algorithmic Framework}

\author{Klaus Frick}
\address{Institute for Mathematical Stochastics\\
University of G{\"o}ttingen\\
Goldschmidtstra{\ss}e 7, 37077 G{\"o}ttingen}
\email{frick@math.uni-goettingen.de}
\thanks{Correspondence to frick@math.uni-goettingen.de }
  
\author{Philipp Marnitz}
\address{Institute for Mathematical Stochastics\\
University of G{\"o}ttingen\\
Goldschmidtstra{\ss}e 7, 37077 G{\"o}ttingen}
\email{marnitz@math.uni-goettingen.de}

\author{Axel Munk}
\address{Institute for Mathematical Stochastics\\
University of G{\"o}ttingen\\
Goldschmidtstra{\ss}e 7, 37077 G{\"o}ttingen\\
 and}
\address{Max Planck Institute for Biophysical Chemistry \\
Am Fa{\ss}berg 11, 37077 G{\"o}ttingen}
\email{munk@math.uni-goettingen.de}

      
                             
\begin{abstract}           
In this paper we are concerned with fully automatic and
locally adaptive estimation of functions in a ``signal $+$ noise''-model where the regression
function may additionally be blurred by a linear operator, e.g. by a
convolution. To this end, we introduce a general class of \emph{statistical
multiresolution estimators} and develop an algorithmic framework for computing those. By
this we mean estimators that are defined as solutions of convex optimization problems with
$\ell_\infty$-type constraints. We employ a combination of the alternating
direction method of multipliers with Dykstra's algorithm for
computing orthogonal projections onto intersections of convex sets and prove
numerical convergence. The capability of the proposed method is illustrated by
various examples from imaging and signal detection.  
\end{abstract}  
       
\maketitle   

\newlength{\imwidth}
\setlength{\imwidth}{12.5cm}

\section{Introduction}\label{intro} 

In numerous applications, the relation of
observable data $Y$ and the (unknown) signal of interest  $u^0$ can be
modeled as an inverse linear regression problem. We shall assume that the data
$Y = \set{Y_{\vec\nu}}$ is sampled on the 
equidistant grid $X = \set{1,\dots,m}^d$, with $m,d\in \N$ and that $u^0\in U$
for some linear space $U$, such as the Euclidean space or a Sobolev class of
functions. Hence the model can be formalized as
\begin{equation}\label{intro:lineqn}
  Y_{\vec\nu} = (Ku^0)_{\vec\nu} + \eps_{\vec\nu},\quad \vec\nu \in X.
\end{equation}
Here we assume that $\eps = \set{\eps_{\vec\nu}}_{\vec\nu \in X}$ are independent
and identically distributed r.v. with $\E{\eps_{\vec\nu}} = 0$ and
$\E{\eps_{\vec\nu}^2} = \sigma^2>0$ (white noise). Moreover,  $K:U\ra (\R^{m})^d$
denotes a linear operator that encodes the functional relation between the
quantities that are accessible by experiment and the underlying signal.
Often the operator $K$ does not have a continuous inverse (or its inverse is
ill-conditioned in a discrete setting, where $K$ is a matrix), that is
estimation of $u^0$ given the data $Y$ is an \emph{ill-posed problem}. As a consequence, estimators for $u^0$
can not be obtained by merely applying the inverse of $K$  to an
estimator of $Ku^0$, in general.  Instead, more sophisticated \emph{statistical
regularization}  techniques have to be employed that, loosely speaking, are
capable of simultaneously inverting $K$ and solving the regression problem.

The application we primarily have in mind is the reconstruction of low-dimensional
 signals (e.g. images) $u^0$ which are presumed to
 exhibit a strong neighborhood structure as it is characteristic for imaging or
 signal detection problems. These neighborhood
relations are often modeled by prior smoothness or structural assumptions on
$u^0$ (e.g. on the texture of an image).

The aim of this paper is twofold. First, we will introduce the broad class of
\emph{statistical multiresolution estimators (SMRE)}. We claim that numerous
regularization techniques, that were recently  proposed for different problems
in various branches of applied mathematics and statistics, can be considered as
special cases of these. Among others, this includes the \emph{Dantzig selector}
(see \cite{frick:BicRitTsy09,frick:CanTao07,frick:JamRadLv09} and references
therein) that was recently proposed in the context of high dimensional
statistics. Our prior focus, however, will be put on imaging problems and it
will turn out that the aforementioned neighborhood relations can be modeled
within our SMRE framework in a straightforward manner. This will result in
\emph{locally adaptive} and \emph{fully automatic} image reconstruction methods.

The high intrinsic structure of the signals that are typically under
consideration in imaging is in contrast to the usual situation in
high-dimensional statistics. Here $u^0$ is usually assumed to be unstructured
but to have a sparse representation with respect to some basis of $U$  (cf.
\cite{frick:Tib94, frick:CheDonSau01, frick:CanTao07}). Consequently, the
consistent estimation of $u^0$ is realized by minimizing a regularization functional which
fosters sparsity, such as the $\ell_1$-norm of the coefficients, subject to an
$\ell_\infty$-constraint on the coefficients of the residual, i.e.
\begin{equation}\label{intro:dantzig}
\inf_{u\in U}\norm{u}_1\quad\text{ s.t. }\quad \norm{K^*(Y - Ku)}_\infty\leq q.
\end{equation}
In order to apply this approach for image reconstruction, two modifications
become necessary: Often one aims to minimize other regularization functionals
such as the total variation semi-norm (cf.
\cite{frick:MamGee97,frick:MohBerGolOsh11})  or Sobolev norms, say. Hence, we
suggest to replace the $\ell_1$-norm in \eqref{intro:dantzig} by a general
convex functional $J$ that models the smoothness or texture information of
signals or images (cf. \cite{frick:AujAubBlaCha05,frick:Mey01}). Furthermore, we
relax the $\ell_\infty$-constraint such that neighborhood relations of the image
can be taken into account. This generalizes the Dantzig selector to
this task in a natural way and obviously increases estimation efficientcy.
As we will layout in Paragraph \ref{intro:algo}, this requires new algorithms to compute efficiently the resulting large scale optimization problem.

\subsection{Statistical Multiresolution Estimation}\label{intro:smre} 

We will now introduce the announced class of estimators. To this end, let $\S$
be some index set and $\mathcal{W} = \set{\omega^S~:~S\in\S}$ be a set of given
weight-functions on the grid $X=\set{1,\ldots,m}^d$.
A \emph{statistical multiresolution estimator (SMRE)} (or \emph{generalized
Dantzig selector}), is defined as a solution of the constrained optimization
problem
\begin{equation}\label{intro:smreeqn}
  \inf_{u\in U} J(u)\quad\text{ s.t. }\quad \max_{S\in\S} 
  \abs{\sum_{\vec\nu \in X}
  \omega^S_{\vec\nu}\left(\Lambda(Y-Ku)\right)_{\vec\nu}} \leq q.
\end{equation}
Here, $J:U\ra \R$ denotes a regularization functional that incorporates a priori
knowledge on the unknown signal $u^0$ (such as smoothness) and
$\Lambda:(\R^m)^d\ra(\R^m)^d$ a possibly non-linear transformation. The
constant $q$ can be considered as a universal \emph{regularization parameter}
that governs the trade-off between regularity and data-fit of the reconstruction. In
most practical situations $q$ is chosen to be the $\alpha$-quantile $q_\alpha$ of the
\emph{multiresolution (MR) statistic} $T(\eps)$, where $T:(\R^m)^d\ra \R$
encodes the inequality constraint in \eqref{intro:smreeqn}, i.e. 
\begin{equation}\label{intro:mrstateqn}
  T(v) = \max_{S\in\S} \abs{\sum_{\vec\nu \in X}
  \omega^S_{\vec\nu}\left(\Lambda(v)\right)_{\vec\nu}},\quad v\in(\R^m)^d.
\end{equation}
To this end, we assume the distribution of $T(\eps)$ to be (approximately)
known. This can either be obtained by simulations or in some cases the limiting
distribution can even be derived explicitly. The regularization parameter
$q$ then admits a sound statistical interpretation: Each solution $\hat
u_\alpha$ of \eqref{intro:smreeqn} satisfies
\begin{equation*}
  \Prob\left( J(\hat u_\alpha) \leq J(u^0)\right) \geq \alpha
\end{equation*}
i.e. the estimator $\hat u_\alpha$ is \emph{more regular} (in terms of $J$)
than $u^0$ with a probability of at least $\alpha$.  To see this simply observe
that the true signal $u^0$ satisfies the constraint in \eqref{intro:smreeqn}
with probability at least $\alpha$. 

For a given estimator $\hat u$ of $u^0$, the set $\mathcal{W}$ is assumed to be
rich enough in order to catch all relevant non-random signals that are visible
in the residual $Y - K\hat u$. Then, the average function 
\begin{equation} \label{intro:mean} 
  \mu_{S}(v) = \abs{\sum_{\vec\nu \in X}\omega^S_{\vec\nu} 
  \left(\Lambda(v)\right)_{\vec\nu}}
\end{equation}
evaluated at $v = Y-K\hat u$ is supposed to be significantly larger than $q$
for at least one $\omega\in \mathcal{W}$, whenever $Y-K\hat u$ fails to resemble
white noise. Put differently, the MR-statistic  $T(Y-K\hat u)$
is bounded by $q$, whenever $Y - K\hat u$ is accepted as white noise according
to the \emph{resolution} provided by $\mathcal{W}$. In fact, this is a key
observations that reveals numerous potential application areas of the
estimation method \eqref{intro:smreeqn}. The examples we have in mind are
mainly from \emph{statistical signal detection and imaging}, where the index set
$\S$ is typically chosen to be an overlapping (redundant) system of subsets of the grid $X$ and
$\omega^S$ is the normalized indicator function on $S\in\S$. Consequently the inequality constraint in \eqref{intro:smreeqn} guarantees that
the residual resembles white noise on all sets $S\in\S$. In other words, the
SMRE approach in \eqref{intro:smreeqn} yields a
reconstruction method that \emph{locally adapts the amount of regularization}
according to the underlying image features. We illustrate this in Section
\ref{appl} by various examples. 

Summarizing, the optimization problem in \eqref{intro:smreeqn} amounts to
choose the most parsimonious among all estimators $\hat u$ for which the
residual $Y - K\hat u$ resembles white noise according to the statistic $T$.
If $Y-K\hat u$ contains some non randon
 signal, i.e. $T(Y - K\hat u)$ is likely to be larger than $q$ and $u$ happens
 to lie outside the admissible domain of \eqref{intro:smreeqn}. Thus,  the
 multi-resolution constraint prevents too parsimonious reconstructions due to the minimization of $J$.   
   
\subsection{Algorithmic Challenges and Related Work}

\subsubsection{Multiresolution Methods} SMREs and related MR statistics
have recently been studied in various contexts.  We give a brief (but
incomplete) overview.

Classical MR statistics are obtained from the general form in
\eqref{intro:mrstateqn} by setting $U = (\R^m)^d$ and $\Lambda = \id$.
Moreover, one considers the system $\mathcal{W}$ to contain indicator
functions on cubes. To be more precise, define the index set $\S$ to be the
system of all $d$-dimensional cubes in $X$ and set $\omega^S =
 \chi_S\slash \sqrt{\#S}.$ Then, the MR-statistic in \eqref{intro:mrstateqn}
 reduces to
\begin{equation*}
  T(v) = \max_{S\in\S}\frac{1}{\sqrt{\#S}}\abs{\sum_{\vec\nu \in S}
  v_{\vec\nu}}.
\end{equation*}
This statistic was introduced in \cite{frick:SieYak00} (called
scanning statistic there) in order to detect a signal against a noisy background. It
was shown in \cite{frick:KabMun08} 
that 
\begin{equation*}
  \lim_{m\ra\infty} \frac{T(\eps)}{\sqrt{2d\log m}} = \sigma\quad\text{ a.s.}
\end{equation*}
If the system $\mathcal{S}$ is reduced to the set of all \emph{dyadic} squares,
then it was proved in \cite{frick:HotMarStiDavKabMun12} that (after suitable
transformations) $T$ also converges weakly to the Gumbel
distribution. There, the authors also established a method for locally
adaptive image denoising employing linear diffusion equations with spatially
varying diffusivity. SMREs \eqref{intro:smreeqn} have been studied
recently for the case $d=1$ in \cite{frick:DavKovMei09} and
\cite{frick:BoyKemLieMunWit09}, where total-variation penalty and the number of
jumps in piecewise constant regression were considered as regularization functional $J$,
respectively. In \cite{frick:FriMarMun10} consistency and convergence rates
for SMREs have been studied in a general Hilbert space setting.

SMREs with squared residuals, that is $\Lambda(v)_{\vec\nu} =
v_{\vec\nu}^2$, yield another class of estimators that have attracted much
attention. Above all, the situation where $\S$ consists
of the single set $X$ and $\omega^X$ is chosen to be the constant $1$ function is of special
interest, since then \eqref{intro:smreeqn} reduces to the \emph{penalized least
square estimation}. In particular \eqref{intro:smreeqn} then can be rewritten
into
\begin{equation}\label{intro:penleastsquare}
  \inf_{u\in U} J(u) + \lambda \sum_{\vec\nu \in X} (Ku-Y)_{\vec\nu}^2
\end{equation}
for a suitable multiplier $\lambda > 0$. If $J(u) = \norm{u}_1$ the
LASSO estimator will result (cf. \cite{frick:Tib94}). Recently, also non-trivial
choices of $\S$ were considered. In \cite{frick:BerCasRouSol03} $\S$ is chosen to consist of a partition of $G$ which is obtained beforehand by a Mumford-Shah segmentation.
In \cite{frick:DonHinRin11}, a subset $S\subset X$ is fixed and afterwards $\S$
is defined as the collection of all translates of $S$.

In \cite{frick:DueSpo01} MR-statistics are used for shape-constrained
estimation based on testing qualitative hypothesis in nonparametric regression
for $d=1$. Here, the weight functions $\omega^S$ incorporate qualitative features such as monotonicity or concavity.
Similarly, MR-statistics are used in \cite{frick:DueWal08} in order to detect
locations of local increase and decrease in density estimation. Much in the same
spirit is the work in \cite{frick:DueJoh04} where multiscale sign tests are
employed for computing confidence bands for isotonic median curves.
 
As mentioned previously, the \emph{Dantzig-selector} \cite{frick:CanTao07}
is also covered by the general SMRE framework in \eqref{intro:smreeqn}. To see
this, set $U = \R^p$ (with typically $p\gg m$), $\Lambda = \id$ and define the weights
\begin{equation*}
  \omega^S = K \chi_{S},\quad S\in \S.
\end{equation*} 
Then, each solution of \eqref{intro:smreeqn} can be considered as a generalized
Dantzig selector. The matrix $K\in\R^{m\times p}$ in this context is usually
interpreted as \emph{design matrix} of a high dimensional linear model. The
classical Dantzig selector as introduced in \cite{frick:CanTao07} then results
in the special case where $\S$ only consists of single-elemented subsets of
$\set{1,\ldots,p}$ and $J$ is chosen to be the $\ell_1$-regularization
functional
\begin{equation*}
  J(u) = \norm{u}_1 = \sum_{i=1}^p \abs{u_i}.
\end{equation*}
Hence LASSO and Dantzig selector are uni-scale estimators which take into
account the largest ($\S = \set{X}$) and smallest ($\S$ consists of all
singletons in $\set{1,\ldots,p}$) scales, respectively. In this sense, they
constitute two extreme cases of SMRE.
          
\subsubsection{Algorithmic Challenges}\label{intro:algo} From a computational
point of view, computing an SMRE amounts to solve the \emph{constrained
optimization problem} \eqref{intro:smreeqn} which can be rewritten into
\begin{equation}\label{intro:smreeqnp}
  \inf_{u\in U} J(u)\quad\text{ s.t. }\quad \mu_{S}(Y - Ku)
  \leq q, \; \forall(S\in\mathcal{S}).  
\end{equation}  
We note that in practical applications the number of constraints in
\eqref{intro:smreeqnp}, that is the cardinality of the index set $\mathcal{S}$,
can be quite large (in Section \ref{appl:denoising} denoising of a $512\times
512$ image results in more than $6$ million inequalities). Moreover, the inequalities
(even for the simplest case where $\Lambda = \id$) are mutually correlated. Both of these facts turn
\eqref{intro:smreeqnp} into a numerically challenging problem and standard
approaches (such as interior point or conjugate gradient methods) perform far from satisfactorily.
  
The authors in \cite{frick:BerCasRouSol03,frick:DonHinRin11,frick:HotMarStiDavKabMun12}
approach the numerical solution of \eqref{intro:smreeqnp}  by means of an analogon of 
\eqref{intro:penleastsquare} with spatially dependent multiplier
$\lambda\in(\R^m)^d$,  i.e. 
\begin{equation*}
  \inf_{u\in U} J(u) + \sum_{\vec\nu \in X}\lambda_{\vec\nu}
  (Ku-Y)_{\vec\nu}^2.
\end{equation*} 
Starting from a (constant) initial parameter $\lambda = \lambda_0$, the
parameter $\lambda$ is iteratively adjusted by increasing it in regions
which were poorly reconstructed before according to the MR-statistic $T$. 
This approach strongly depends on the special structure of $\S$ that allows a
straightforward identification of each set $S\in\S$ with a unique point in the
grid $X$. Put differently, 
it is not clear how to modify this paradigm  in order to solve
\eqref{intro:smreeqn} for highly redundant systems $\S$ as we have it in mind. 

Recently a general algorithmic framework was introduced in
\cite{frick:BecCanGra10} for the solutions of large-scale convex cone problems
\begin{equation*}
  \inf_{u\in U} J(u) \quad\text{ s.t. }\quad Ku-Y \in \mathcal{K}
\end{equation*}
where $\mathcal{K}$ is a convex cone in some Euclidean space. The approach was
realized in the software package \emph{Templates for First-Order Conic Solvers
(TFOCS)} \footnote{available at \url{http://tfocs.stanford.edu/}}. The above
formulation is very general and in order to
recover \eqref{intro:smreeqnp} one has to consider the cone
\begin{equation*}
  \mathcal{K} = \set{(v,q)\in (\R^m)^d\times \R ~:~ 
  \abs{\sum_{\vec\nu\in S} \Lambda(v)_{\vec\nu}} \leq q\;\forall(S\in\S)} 
\end{equation*}  
The approach in \cite{frick:BecCanGra10} employs the dual formulation of the
problem 
\begin{equation*}
\inf_{\xi \in V} J^*(K^* v) + \inner{Y}{v}\quad\text{ s.t. } v \in \mathcal{K}^*
\end{equation*}
which involves the computation of the dual cone $\mathcal{K}^*$ ($J^*$ denotes
the Legendre-Fenchel dual of $J$). This approach is particularly appealing for
the uni-scale Dantzig selector since in this situation the cone $\mathcal{K}$
coincides with the epi-graph of the $\ell^\infty$-norm and hence its dual cone is
straightforward to compute (it is the epi-graph of the $\ell^1$-norm). As it is
argued in \cite{frick:BecCanGra10}, this approach is capable of computing
Dantzig selectors for large scale problems in contrast to previous approaches
such as standard linear programming techniques
\cite{frick:CanTao07} or homotopy methods such as DASSO \cite{frick:JamRadLv09}
or \cite{frick:Rom08}. As the authors stress, their approach
works well in the case when $\mathcal{K}$ is the epi-graph of a
norm for which the projections onto $\mathcal{K}^*$ are tractable and
computationally efficient. However, for the applications we have in mind (such
as locally adaptive imaging reconstruction), the approach in
\cite{frick:BecCanGra10} is only of limited use: In contrast to the
aforementioned epi-graphs, the large number of (strongly dependent) constraints
in \eqref{intro:smreeqnp} brings about a cone $\mathcal{K}$ that on the one hand
exhibits a tremendous amount of faces compared to the dimension of the image
space $\dim(H) = md$ and that on the other hand is no longer symmetric w.r.t. to
the $q$-axis. Both of these facts turn the computation of dual cone
$\mathcal{K}^*$ (or the projections onto it) into a most
 challenging problem, even in the simplest case when $\Lambda$ is linear.

The aim of this paper is to develop a general algorithmic framework that makes
solutions of \eqref{intro:smreeqnp} numerically accessible for many
applications. In order to do so we propose to introduce a slack variable in
\eqref{intro:smreeqnp} and then use the \emph{alternating direction method of
multipliers}, an Uzawa-type algorithm that decomposes problem
\eqref{intro:smreeqnp} into a $J$-penalized least squares problem for the primal variable and a orthogonal
projection problem on the feasible set of \eqref{intro:smreeqnp} for the slack
variable. This approach has the appealing effect that once an implementation for
the projection problem is established, different regularization functionals $J$
can easily be employed without changing the backbone of the algorithm. Our work
is much in the same spirit as \cite{frick:LuPonZha10}, which considered an
alternating direction method for the computation of the Dantzig selector
recently. In this case the computation of the occurring orthogonal
projections are available in closed form, whereas in our applications this is
not the case due to the aforementioned dependencies.
  
In order to tackle the orthogonal projection problem we employ Dykstra's
projection method \cite{frick:BoyDyk86} which is capable of computing the
projection onto the intersection of convex bodies by merely using the
individual projections onto the latter. The efficiency of the proposed
method hence increases considerably if the index set $\S$ can be decomposed into
``few" partitions that contain indices of mutually independent inequalities in
\eqref{intro:smreeqnp}.  In particular,  by this approach we will be able to
compute classical SMRE  (as introduced in \cite{frick:DavKovMei09,
frick:FriMarMun10}) in $d=2$  space dimensions which to our knowledge has never
been done so far.  This puts us into the position to study the performance
of such estimators compared with  other benchmark methods in locally adaptive
signal recovery (such as \emph{adaptive weights smoothing} cf.
\cite{frick:PolSpo00}). As it will turn out in Section \ref{appl} it will outperform these visually as well as quantitatively. 

\subsection{Organization of the Paper}
 
The paper is organized as follows: In Section \ref{impl} we introduce a general
algorithmic approach for computing SMREs. We will rewrite \eqref{intro:smreeqnp}
into a linearly constrained problem and compute a saddle point of the
corresponding augmented Lagrangian by the alternating direction method of
multipliers in Paragraph \ref{impl:deco}. Under quite general
assumption, we prove convergence of the algorithm in Theorem \ref{impl:alaconv} and give some
qualitative estimates for the iterates in Theorem \ref{impl:alaconvcor}. One of
the occurring minimization steps amounts to the computation of an orthogonal
projection onto a convex set in Euclidean space. In Paragraph \ref{impl:proj},
this problem will be tackled by means of Dykstra's projection algorithm
introduced in \cite{frick:BoyDyk86}. Finally, we illustrate the performance of
some particular instances of SMREs in Section \ref{appl}: we study
problems in nonparametric regression, image denoising and deconvolution of
fluorescence microscopy images and compare our results to other methods by means
of simulations. 

\section{Computational Methodology}\label{impl}

In this section we will address the question on how to solve the linearly
constrained optimization problem \eqref{intro:smreeqnp}. After discussing some
notations and basic assumptions in Subsection \ref{review:assnot}, we will
reformulate the problem in Paragraph \ref{impl:deco} such that the alternating
direction method of multipliers (ADMM), a Uzawa-type algorithm,  can be employed
as a solution method. As an effect, the task of computing a solution of \eqref{intro:smreeqnp} is
replaced by alternating 
\begin{enumerate}[i)] 
  \item solving  an unconstrained
penalized least squares problem that is \emph{independent of the MR-statistic
$T$} and 
\item computing the orthogonal projection on a convex set in Euclidean
space that is \emph{independent of $J$}.
\end{enumerate}
This reveals an appealing modular nature of our approach: The regularization
functional $J$ can easily be replaced once a method for the projection problem
is settled. For the latter we will propose an iterative projection algorithm in
Paragraph \ref{impl:proj} that was introduced by Boyle and Dykstra in
\cite{frick:BoyDyk86}.

\subsection{Basic Assumptions and Notation}\label{review:assnot}

From now on, $X$ will stand for the $d$-dimensional grid
$\set{1,\ldots,m}^d$ and agree upon  $H = \R^X \simeq (\R^m)^d$ being the space
of all real valued functions $v:X\ra \R$. Moreover, we assume that $\S$ denotes
some index set and that $\mathcal{W} = \set{\omega^S~:~S\in\S}$ is a collection
of elements in $H$. For two elements $v,w \in H$ we will use the standard inner
product and norm
\begin{equation*}
  \inner{v}{w} = \sum_{\vec\nu \in X} v_{\vec\nu} w_{\vec\nu}\quad \text{ and
  }\quad \norm{v} = \sqrt{\inner{v}{v}}
\end{equation*}
respectively. Next, we assume that $\Lambda:H\ra H$ is continuous such that
$\Lambda(0) = 0$ and that for all $S\in\S$ the mapping
\begin{equation*}
  v\mapsto \inner{\omega^S}{\Lambda(v)}
\end{equation*}
is convex.  With this notation, we can rewrite the average
function in \eqref{intro:mean} in the compact form 
\begin{equation*} 
  \mu_{S}(v) = \abs{\inner{w^S}{\Lambda(v)}}.
\end{equation*}
We note, that it is not restricitve to consider more generaly $\Lambda:H\ra
\R^N$ with arbitrary $N\in\N$. This could e.g. be useful for augmenting the
constraint set of \eqref{intro:smreeqnp} with further constraints of different
type. For the signal and image detection problems as studied in this paper,
however, $\Lambda$ is always a pointwise transformation of the residuals. Hence, 
we will restrict our considerations on the case when $\Lambda:H\ra H$.

Furthermore, we define $U$ to be a separable Hilbert-space with inner product
$\inner{\cdot}{\cdot}_U$ and induced norm $\norm{\cdot}_U$. The operator $K:U\ra
H$ is assumed to be linear and  bounded and the functional $J:U\ra \R$ is
convex and lower semi-continuous, that is
\begin{equation*}
  \set{u_n}_{n\in\N}\subset U\text{ and } \lim_{n\ra\infty} u_n =: u\in U
  \quad\Longrightarrow \quad J(u) \leq \liminf_{n\ra\infty}J(u_n).
\end{equation*}
Recall the definition of the MR-statistic in \eqref{intro:mrstateqn}.
Throughout this paper we will agree upon the following 

\begin{assa}\label{review:assex}
\begin{enumerate}[i)]
  \item For all $y\in H$ there exists $u\in U$ such that
  $T(Ku-y)< q$.
\item For all $y\in H$ and $c\in \R$ the set
\begin{equation*}
  \set{u\in U ~:~ \max_{S\in \S}\mu_{S}(Ku-y)
   + J(u) \leq c}
\end{equation*}
is bounded.
\end{enumerate}
\end{assa}

Under Assumption \ref{review:assex} it follows from standard techniques in convex
optimization, that a solution of \eqref{intro:smreeqnp} exists. As we will
discuss in Section \ref{impl:deco} it even follows that a saddle point of the
corresponding Lagrangian exists (cf. Theorem \ref{impl:kktthm} below). In this
context Assumption \ref{review:assex} i) is often referred to as \emph{Slater's
constraint qualification} and is for instance satisfied if $K(U)$ is dense in
$H$. Moreover, Assumption \ref{review:assex} ii) will be needed in order to
guarantee convergence of the algorithm for computing such a solution, as it is
proposed in the upcoming section. This requirement is fulfilled if $J$ is
coercive i.e.
\begin{equation*}
  \lim_{\norm{u}_U\ra\infty} J(u)= \infty.
\end{equation*} 
In many applications $U$ is some function space and $J$ a gradient
based regularization method, such as the total variation semi-norm (cf. Section
\ref{appl:denoising}). Then a typical sufficient condition for Assumption
\ref{review:assex} ii) is that $K$ does not annihilate constant functions.

\subsection{Alternating Direction Method of Multipliers}\label{impl:deco}

By introducing a slack variable $v\in H$ we rewrite \eqref{intro:smreeqnp} to
the equivalent problem
\begin{equation}\label{impl:linconstr}
  \inf_{u\in U, v\in H} J(u) + G(v)  \quad\text{ subject to }\quad Ku +
  v = Y.
\end{equation}
Here, $G$ denotes the characteristic function on the feasible region
$\mathcal{C}$ of \eqref{intro:smreeqnp}, that is,
\begin{equation}\label{impl:feasible}
  \mathcal{C} = \set{v\in H~:~ \mu_{S}(v) \leq q
  \;\forall(S\in\S)}\quad \text{ and }\quad G(v) =
  \begin{cases}
0 & \text{ if } v \in\mathcal{C} \\
+\infty & \text{ else}.
\end{cases}
\end{equation}
Note that due to the assumptions on $\Lambda$, the set $\mathcal{C}$ is closed
and convex. The technique of rewriting \eqref{intro:smreeqnp} into
\eqref{impl:linconstr} is referred to as the \emph{decomposition-coordination
approach}, see e.g. Fortin \& Glowinski  
\cite[Chap. III]{frick:FG83}. There, Lagrangian multiplier methods are used for solving \eqref{impl:linconstr}. To this end, we recall the
definition of the \emph{augmented Lagrangian} of  Problem \eqref{impl:linconstr},
that is
\begin{equation}\label{impl:lagr}
  L_\lambda (u,v;p) = \frac{1}{2\lambda} \norm{Ku + v - Y}^2 + J(u) + G(v) -
  \inner{p}{Ku + v - Y},\quad \lambda > 0.
\end{equation}
The name stems from the fact that the ordinary Lagrangian
\begin{equation*}
  L(u,v;p) = J(u) + G(v) - \inner{p}{Ku + v - Y}
\end{equation*}
is augmented by the quadratic penalty term $(2\lambda)^{-1} \norm{Ku + v - Y}^2$
that fosters the fulfillment of the linear constraints in \eqref{impl:linconstr}.
The \emph{augmented Lagrangian method} consists in computing a saddle point
$(\hat u, \hat v, \hat p)$ of $L_\lambda$, that is
\begin{equation*}
  L_\lambda(\hat u,\hat v; p) \leq L_\lambda(\hat u, \hat v; \hat p) \leq
  L_\lambda( u,  v; \hat p),\quad \forall\left( (u,v,p)\in U\times H\times
  H\right) 
\end{equation*}
We note that each saddle point $(\hat u, \hat v, \hat p)$ of the augmented
Lagrangian $L_\lambda$ is already a saddle point of $L$ and vice versa and
that in either case the pair $(\hat u, \hat v)$ is a solution of
\eqref{impl:linconstr} (and thus $\hat u$ is a desired solution of
\eqref{intro:smreeqnp}). This follows e.g. from \cite[Chap 3. Thm.
2.1]{frick:FG83}. Sufficient conditions for the existence of saddle points are
usually harder to come up with. Assumption \ref{review:assex} summarizes a
standard set of such conditions.

\begin{thm}\label{impl:kktthm}
  Assume that Assumption \ref{review:assex} holds.  Then, there exists a saddle
  point $(\hat u, \hat v, \hat p)$ of $L_\lambda$.
\end{thm}
\begin{proof}
  According to \cite[Chap. III, Prop. 3.1 and Prop. 4.2]{frick:ET76} 
  a saddle point of $L$ exists, if there  is an element $u_0\in U$ such
  that $G$ is continuous at $Ku_0-Y$ and that
  \begin{equation}\label{impl:coercivity}
       \lim_{\norm{u}_Q\ra\infty} J(u) + G(Ku-Y) =  \infty.
  \end{equation}  
  According to Assumption \ref{review:assex} i) and due to the continuity of
  $\Lambda$ the first requirement is clearly satisfied. Further, the coercivity
  assumption \eqref{impl:coercivity} is a consequence of Assumption
  \ref{review:assex} ii).  
\end{proof}
 
We will use the \emph{Alternating Diretion Method of Multipliers (ADMM) } (cf.
Algorithm \ref{impl:ala}) as proposed in \cite[Chap. III Sec. 3.2]{frick:FG83}
for the computation of a saddle point of $L_\lambda$ (and hence of a solution of
\eqref{intro:smreeqnp}): Successive minimization of the augmented Lagrangian
$L_\lambda$ w.r.t. the first and second variable followed by an explicit step
for maximizing w.r.t. the third variable is performed. Convergence of this
method is established in Theorem \ref{impl:alaconv} which is a generalization of
\cite[Chap. III Thm. 4.1]{frick:FG83}. We note that the proof, as presented in
the Appendix \ref{app} allows for \emph{approximate} solution of the individual
subproblems. For the sake of simplicity, we present the Algorithm in its exact
form.

\begin{algorithm}[!ht]\caption{Alternating Direction Method of Multipliers}
\label{impl:ala}
\begin{algorithmic}
\REQUIRE $Y\in H$ (data), $\lambda > 0$ (step size), $\tau \geq 0$ (tolerance).
\ENSURE $(u[\tau],
v[\tau])$ is an approximate solution of\eqref{impl:linconstr} computed in $k[\tau]$ iteration steps.
\STATE $u_0\leftarrow \vec 0_{U}$ and $v_0 = p_0\leftarrow \vec 0_{H}$
\STATE $r \leftarrow \norm{Ku_0 + v_0 - Y}$ and $k \leftarrow 0$.
\WHILE{$r > \tau$} 
\STATE $k\leftarrow k+1$.
\STATE $v_k  \leftarrow  \tilde v$ where $\tilde v\in \mathcal{C}$
satisfies
\begin{equation}\label{ala:noise}
  \norm{\tilde v - (Y + \lambda p_{k-1} - Ku_{k-1}) }^2 \leq  \norm{v - (Y
  + \lambda p_{k-1} - Ku_{k-1}) }^2  \quad \forall(v\in \mathcal{C}).
\end{equation}
\STATE $ u_k  \leftarrow  \tilde u$ where $\tilde u$ satisfies
\begin{equation}\label{ala:primal}
  \frac{1}{2}\norm{ K\tilde u - (Y + \lambda p_{k-1} - v_k) }^2 + \lambda
  J(\tilde u) \leq \frac{1}{2}\norm{ K u - (Y + \lambda p_{k-1} - v_k) }^2 + \lambda
  J( u) \quad\forall(u\in U).
\end{equation}
\STATE $p_k  \leftarrow  p_{k-1} -  (K u_k + v_k -
Y)\slash \lambda$.%\label{ala:update}
\STATE $r\leftarrow \max(\norm{Ku_k + v_k - Y}, \norm{K(u_k-u_{k-1})})$.
\ENDWHILE
\STATE $u[\tau] \leftarrow u_k$ and $v[\tau] \leftarrow v_k$ and $k[\tau]
\leftarrow k$.
\end{algorithmic}
\end{algorithm}

\begin{thm}\label{impl:alaconv}
  Every sequence $\set{(u_k, v_k)}_{k\geq1}$ that is generated
  by Algorithm \ref{impl:ala} is bounded in $U\times H$ and every weak cluster point is a solution of \eqref{impl:linconstr}. Moreover,
  \begin{equation*}
    \sum_{k\in\N} \norm{Ku_k + v_k - Y}^2 + \norm{K(u_k - u_{k-1})}^2 < \infty.
  \end{equation*}
\end{thm}

\begin{rem}%\label{alaconvrem}
  \begin{enumerate}[i)]
    \item Theorem \ref{impl:alaconv} implies, that each weak cluster
    point of $\set{u_k}_{k\geq1}$ is a solution of \eqref{intro:smreeqnp}.
    In particular, if the solution $u^\dagger$ of \eqref{intro:smreeqnp} is unique
    (e.g. if $J$ is strictly convex), then $u_k \rightharpoonup u^\dagger$.
    \item Note in particular that \eqref{ala:noise} is independent of the
    choice of $J$, while \eqref{ala:primal} is independent of the
    multiresolution statistic being used. This decomposition gives the proposed
    method a neat modular appeal: once an efficient solution method for the
    projection problem \eqref{ala:noise} is established (see e.g. Section
    \ref{impl:proj}), the regularization functional $J$ in \eqref{intro:smreeqn} can easily
    be replaced by providing an algorithm for the penalized least squares
    problem \eqref{ala:primal}. For most popular choices of $J$, problem 
    \eqref{ala:primal} is well studied and efficient computational methods are
    at hand (see \cite{frick:Vog02} for a extensive collection of algorithms and
    \cite{frick:KaiSom05} for an overview on MCMC methods).
  \end{enumerate}
\end{rem}

For a given tolerance $\tau > 0$, Theorem \ref{impl:alaconv} implies that
Algorithm \ref{impl:ala} terminates and outputs approximate solution $u[\tau]$
and $v[\tau]$ of \eqref{impl:linconstr}. However, the breaking condition in
Algorithm \ref{impl:ala} merely guarantees that the linear constraint in
\eqref{impl:linconstr} is approximated sufficiently well. Moreover, we know
from construction that $v[\tau] \in \mathcal{C}$, which implies
$G(v[\tau]) = 0$. So, it remains to evaluate the validity of $u[\tau]$:
  
\begin{thm}\label{impl:alaconvcor} 
Let $(\hat u, \hat v, \hat p) \in U\times H\times H$ be any saddle point of
$L_\lambda$. Moreover, let $\tau > 0$ and $u[\tau]\in U$ be returend by
Algorithm \ref{impl:ala}. Then,
\begin{equation*}
  0\leq J(u[\tau]) -  J(\hat u) - \inner{K^*\hat p}{u[\tau] - \hat u }_U\leq 
  \tau\left( 6\norm{\hat p} + \frac{4\norm{K\hat
  u} + 2\tau}{\lambda} \right) \quad\forall(\tau >
  0).
\end{equation*}
\end{thm} 

The result in Theorem \ref{impl:alaconvcor} shows how the accuracy of the
approximate solution $u[\tau]$  depends on $\tau$. Moreover, it reveals that
choosing a small step size $\lambda$ in Algorithm \ref{impl:ala} possibly 
yields a slow decay of the objective functional $J$. However, it
follows from the definition of $L_\lambda$ in \eqref{impl:lagr} that a small
value for $\lambda$ fosters the linear constraint in \eqref{impl:linconstr}.
  
\begin{cor}\label{impl:alaconvcortwo}
Let the assumtions of Theorem \ref{impl:alaconvcor} be satisfied. Moreover,
assume that $J$ is a quadratic functional, i.e.  $J(u) =
\frac{1}{2}\norm{Lu}_V^2$, where $V$ is a further Hilbert-space and $L:U\supset
D \ra V$ is a linear, densely-defined and closed operator.  Then
\begin{equation*}
\norm{L(u[\tau] - \hat u)} = \bigo(\sqrt{\tau}) \quad\forall(\tau > 0).
\end{equation*}
\end{cor}

\begin{example}[Dantzig selector]%\label{impl:exdantzig}
  As already mentioned in the introduction, SMRE (i.e.\
  finding solutions of \eqref{intro:smreeqn}) reduces to the computation of
  \emph{Dantzig selectors} for the particular setting $d=1$, $U = \R^p$ (with usually
  $p\gg m$) and 
  \begin{equation*}
      J(u) = \norm{u}_{1}.	
  \end{equation*}
  When applying Algorithm \ref{impl:ala}
  the subproblem \eqref{ala:primal} amounts to compute  
  \begin{equation*}
      u_k \in \argmin_{u\in \R^p} \frac{1}{2}\norm{Ku - (Y+ \lambda p_{k-1} -
      v_k)}^2 + \lambda
      \norm{u}_{1}.
  \end{equation*}
 This is the well known \emph{least absolute shrinkage and selection operator
 (LASSO)} estimator \cite{frick:Tib94}.  For the classical Dantzig selector,
 one chooses $\S = \set{1,\ldots,p}$ and defines for $S\in\S$ the weight
 $\omega^S = K \chi_{\set{S}}$. Hence, the subproblem \eqref{ala:noise} in this  case
 consists in the orthonormal projection  of $Y_k = Y+ \lambda p_{k-1} -
 Ku_{k-1}$ onto the set  
 \begin{equation*}
 	\mathcal{C} = \set{v\in\R^m~:~ \abs{\sum_{1\leq j\leq m} \omega^S_j v_j}\leq
 	q\text{ for }1\leq S\leq p}.
 \end{equation*}
 The implications of Theorem \ref{impl:alaconvcor} in the present case are in
 general rather weak. If the saddle point $\hat u$ is known to be
 $S$-sparse and when $K$ restricted to the support of  $\hat u$ is injective,
 then it can be shown that $\abs{u[\tau] - \hat u}_{1} = \bigo(\tau)$.
 
 We finally note that for this particular situation a slightly
 different decomposition than proposed in \eqref{impl:deco} is favorable. To
 be more precise, define $\tilde K = K^T K$ and $\tilde Y = K^T Y$ and consider 
 \begin{equation*}
    J(u) + \tilde G(v)\quad\text{ subject to }\quad \tilde K u - v = \tilde Y.
 \end{equation*} 
 where $\tilde G$ is the characteristic function on the set $\set{v\in H~:~
 \norm{v}_\infty \leq q}$. Algorithm \ref{impl:ala} applied to this modified
 decomposition then results in the ADMM as introduced in
 \cite{frick:LuPonZha10}. In this case the projection in step \eqref{ala:noise}
 has a closed from.  
 \end{example}

\subsection{The Projection Problem}\label{impl:proj}

Algorithm \ref{impl:ala} resolves the constrained convex optimization problem
\eqref{intro:smreeqnp} into a quadratic program \eqref{ala:noise} and an
unconstrained optimization problem \eqref{ala:primal}. The quadratic program
\eqref{ala:noise} in the $k$-th step of Algorithm \ref{impl:ala} can be written
as a projection:
\begin{equation}\label{impl:quadprob}
  \inf_{v\in H}\norm{v - Y_k}^2 
  \quad\text{ subject to } \quad\mu_{S}(v) \leq q\;\forall(s\in \S)
\end{equation}
where $Y_k = Y + \lambda p_{k-1} -Ku_{k-1}$. We reformulate the side conditions
to
\begin{equation}\label{impl:sidecond}
  v\in \mathcal{C} = \bigcap_{S\in \S} C_{S} \quad \text{ where
  }
  \quad C_{S} = \set{v \in H: \mu_{S}(v) \leq q}.
\end{equation}
The sets $C_{S}$ are closed and convex and
problem \eqref{impl:quadprob} thus amounts to compute the projection
$P_{\mathcal{C}}(Y_k)$ of $Y_k$ onto the intersection
$\mathcal{C}$ of closed and convex sets. According to this interpretation, we use Dykstra's projection algorithm
as introduced in \cite{frick:BoyDyk86} to solve \eqref{impl:quadprob}. This
algorithm takes an element $v \in H$ and convex sets $D_1,\ldots,D_M \subset H$
as arguments. It then creates a sequence converging to the projection of $v$ onto
the intersection of the $D_j$ by successively performing projections onto
individual $D_j$'s. To this end, let $P_D(\cdot)$ denote the projection onto $D
\subset H$ and $S_D = P_D - \id$ be the corresponding projection step. Dykstra's
method is summarized in Algorithm \ref{impl:dyk}.

\begin{algorithm}[!ht]\caption{Dykstra's Algorithm}\label{impl:dyk}
\begin{algorithmic}
\REQUIRE $h \in H$ (data), $D_1,\ldots,D_M \subset H$ (closed and convex
sets) 
\ENSURE A sequence $\set{h_k}_{k\in\N}$ that converges strongly to
$P_\mathcal{D}(h)$ where $\mathcal{D} = \bigcap_{j=1,\ldots,M} D_j$

\STATE $h_{0,0} \leftarrow h$
\FOR{$j=1$ to $M$}
\STATE $h_{0,j} \leftarrow P_{D_j}(h_{0,j-1})$ and $Q_{0,j} \leftarrow
S_{D_j}(h_{0,j-1})$
\ENDFOR
\STATE $h_1\leftarrow h_{0,M}$ and $k \leftarrow 1$
\FOR{$k \geq 1$}
\STATE $h_{k,0} \leftarrow h_k$
\FOR{$j=1$ to $M$}
\STATE $h_{k,j} \leftarrow P_{D_j}(h_{k,j-1} - Q_{k-1,j})$ and $Q_{k,j}
\leftarrow S_{D_j}(h_{k,j-1} - Q_{k-1,j})$
\ENDFOR
\STATE $h_{k+1} \leftarrow h_{k,M}$ and $k \leftarrow k + 1$
\ENDFOR
\end{algorithmic}
\end{algorithm}

A natural explanation of the algorithm in a primal-dual framework as well as
a proof that the sequence $\set{h(M,k)}_{k\in\N}$ converges to
$P_{\mathcal{D}}(h)$ in norm can be found in \cite{frick:GM89,frick:DeuHun94}.
For the case when $\mathcal{D}$ constitutes a polyhedron even explicit error
estimates are at hand (cf. \cite{frick:Xu00}):

\begin{thm}
  Let $\set{h_k}_{k\in\N}$ be the sequence generated by Algorithm \ref{impl:dyk}
  and $P_\mathcal{D}(h)$ be the projection of the input $h$ onto $\mathcal{D}$. Then
  there exist constants $\rho > 0$ and $0\leq c< 1$ such that for all $k\in\N$
  \begin{equation*}
    \norm{h_k - P_\mathcal{D}(h)} \leq \rho c^k.
  \end{equation*}
\end{thm}

\begin{rem}
  \label{impl:dyk_conv}
  The constant $c$ increases with the number $M$ of convex sets which
  intersection form the set $D$ that $h$ is to be projected on. The convergence
  rate therefore improves with decreasing $M$. For further details and estimates
  for the constants $\rho$ and $c$, we refer to \cite{frick:Xu00}.
\end{rem}

Note that application of Dykstra's algorithm is particularly appealing if the
projections $P_{D_j}$ can be
easily computed or even stated explicitly, as it is the case in the following
examples.

\begin{example}\label{impl:projexpllin}
  Assume that $\Lambda = \id$. Then the sets $C_{S}$ are the
  rectangular cylinders
  \begin{equation*}
      C_{S} = \set{v\in H~:~ \abs{\inner{\omega^S}{v}}\leq q}.
  \end{equation*}
  The projection can therefore be explicitly computed as
  \begin{equation*}%\label{impl:singleprojectionlin}
  P_{C_{S}}(v) = \begin{cases}
v -
\sign\left(\inner{\omega^S}{v}\right)\frac{\omega^S}{\norm{\omega^S}}\left(\frac{\abs{\inner{\omega^S}{v}}}{\norm{\omega^S}}
- q \right) &
\text{ if } 
\mu_{S}(v)  >
q\\ v & \text{ else} \end{cases}.
  \end{equation*}
  The left image in Figure \ref{impl:figadmissible} depicts an example for
  $\mathcal{C}$ for $H = \R^2$. For a detailed geometric interpretation of the
  MR-statistic we also refer to \cite{frick:Mil08}.  
  \end{example}
     
\begin{example}\label{impl:projexplquad}
  Assume that $\Lambda(v)_{\vec\nu} = v_{\vec\nu}^2$. Then, it follows that $v\mapsto \inner{\omega^S}{\Lambda(v)}$ is convex if
  and only if $\omega^S_{\vec\nu}\geq 0$ for all $\vec\nu \in X$. In this
  case, the sets $C_{S}$ are elliptic cylinders
  \begin{equation*}
      C_{S} = \set{v\in H~:~ \sum_{\vec\nu \in X} \omega^S_{\vec\nu}
      v_{\vec\nu}^2\leq q}.
  \end{equation*}    
  Moreover, if $\omega^S_{\vec\nu} \in \set{0,1}$ for all $\vec\nu \in X$, then
  the projection $P_{C_S}$ can be explicitly computed as
  \begin{equation*}%\label{impl:singleprojectionquad}
  P_{C_{S}}(v) = \begin{cases}
     \frac{q}{\inner{\omega^S}{\Lambda(v)}} v \chi_{\set{\omega^S = 1}} + v
     \chi_{\set{\omega^S=0}} &
\text{ if }
\mu_{S}(v)  >
q\\ v & \text{ else} \end{cases}.
  \end{equation*}
  The right image in Figure \ref{impl:figadmissible} depicts an example of
  $\mathcal{C}$ for $H = \R^2$.
\end{example}

\begin{figure}[h!]  
\begin{center} 
{\scriptsize
\psfrag{g}{$Y_k$}
\psfrag{pg}{$\text{P}_{\mathcal{C}}(Y_k)$}
\psfrag{C}{$\mathcal{C}$}
\psfrag{Cs}{$C_S$}    
\psfrag{w1}{$\frac{\omega^S}{\norm{\omega^S}}$}  
\psfrag{w}{$q\frac{\omega^S}{\norm{\omega^S}}$}  
\includegraphics[width =
0.4\imwidth]{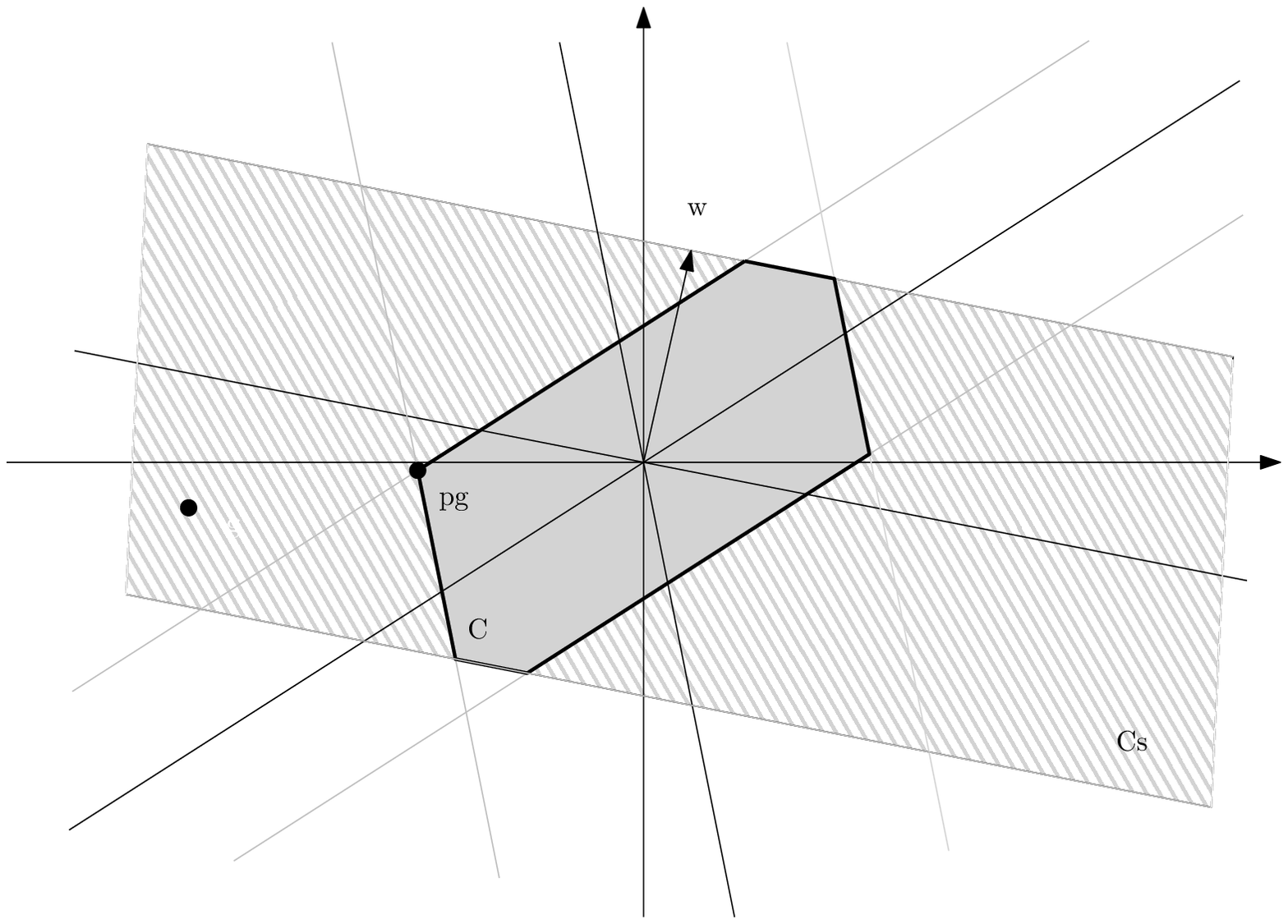}
\hspace{0.05\imwidth}           
\includegraphics[width =
0.4\imwidth]{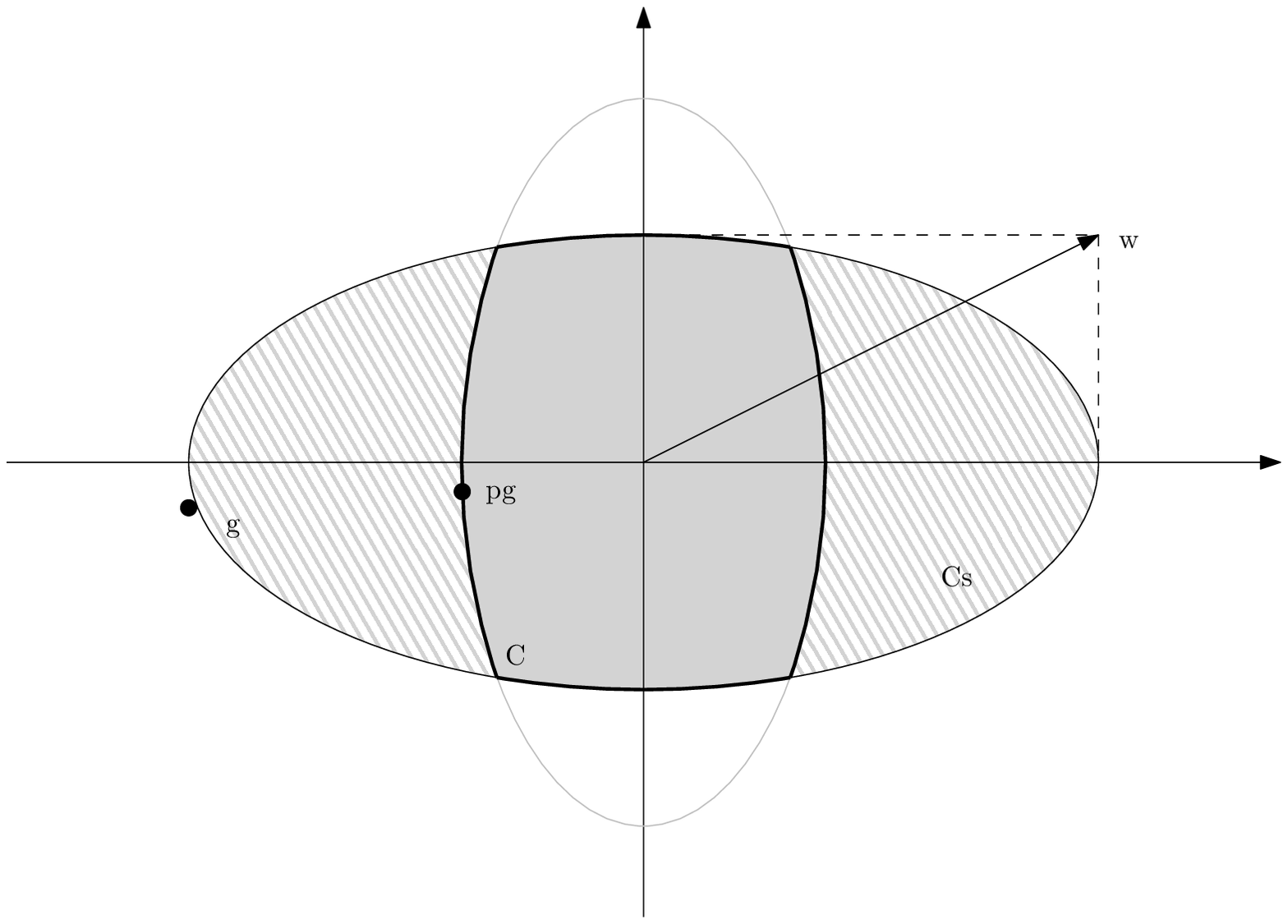}}
\caption{Admissible set $\mathcal{C}$ for the projection problem
\eqref{impl:quadprob} as in Example \ref{impl:projexpllin} (left) and Example
\ref{impl:projexplquad} (right).}\label{impl:figadmissible}
\end{center}
\end{figure} 
 
A first approach to use Dykstra's algorithm
to solve \eqref{impl:quadprob} is to set $M = \#\S$ and
identify $D_j$ with $C_{S}$ for all $j=1,\ldots,M$.  In view of Remark
\ref{impl:dyk_conv}, however, it is clearly desirable to decrease the number $M$
of convex sets that enter Dykstra's algorithm. In order to do so, we take a
slightly more sophisticated approach than the one just presented. We partition
the set $\S$ into $\S_1,\ldots,\S_M$ such that for all $S\not=\tilde S \in
  \S_j$
\begin{equation}\label{impl:orthogonal}
  \omega^S \bot\, \omega^{\tilde S} \quad\text{ and }\quad
  \frac{\partial}{\partial v_{\vec\nu}}\Lambda(\cdot)_{\tilde{\vec\nu}} \equiv
  0,\forall(\vec\nu\in S, \tilde{\vec\nu}\in \tilde S)
\end{equation}
and regroup $\set{C_{S}}_{s\in \S}$ into $\set{D_1,\ldots,D_M}$ with
\begin{equation}\label{impl:regroup} 
  D_j = \bigcap_{S \in \S_j} C_{S}.
\end{equation}
Given the projections $P_{C_{S}}$,  the projection onto $D_j$ can be easily
computed: For $v\in H$ identify the set
\begin{equation*}
  V_j = \set{S\in \S_j : \mu_{S}(v) > q}
\end{equation*}
of indices for which $v$ violates the side condition \eqref{impl:sidecond} and
set
\begin{equation*}%\label{impl:projection}
  P_{D_j}(v) = v- \sum_{S \in V_j} \left(P_{\mathcal{C}_{S}} -
  \id\right) v
\end{equation*}

To keep $M$ small, we choose $\S_1 \subset \set{1,\ldots,N}$ as
the biggest set such that \eqref{impl:orthogonal} holds for all
$S,\tilde S \in \S_1$. We then choose $\S_2 \subset \S \setminus
\S_1$ with the same property and continue in this way until all indices are
utilized. While this procedure does not necessarily result into $M$ being minimal with the desired
property, it still yields a distinct reduction of $N$ in many practical
situations. We will illustrate this approach for SMREs in imaging in
Section \ref{appl}. 

% \begin{example}\label{impl:exdantzigcont}[cont. Example \ref{impl:exdantzig}]
%   The classical Dantzig selector uses $\Lambda = \id$ and the index set $\S =
%   \set{1,\ldots,p}$ with the weight functions $\omega^S = K\chi_{\set{S}}$. Then,
%   clearly the second condition in \eqref{impl:orthogonal} is satisfied and
%   \begin{equation*}
%       \omega^S\bot\,\omega^{\tilde S}\Leftrightarrow
%       \inner{K\chi_S}{K\chi_{\tilde S}} = 0.
%   \end{equation*}
%   Thus, Dykstra's projection algorithm is likely to perform better if the
%   number of clusters of pairwise orthogonal rows in the design matrix
%   $K\in\R^{p\times m}$ decreases. The number of such clusters is
%   bounded from below by $p\slash m$ which can be considerably large and hence
%   the performance of Dykstra's algorithm may not be optimal.
% \end{example} 

%% 
%%
\section{Applications}\label{appl}

In this section we will illustrate the capability of Algorithm \ref{impl:ala} for
computing SMREs in some practical situations: in Section \ref{appl:reg}
we will study a simply one-dimensional regression problem as it was also studied
in \cite{frick:DavKovMei09}, yet with a different penalty function $J$. In
Section \ref{appl:denoising} we illustrate how SMREs performs in image
denoising. In both cases we compare our results to other methods. Finally, we
will apply the SMRE technique to the problem of image deblurring in confocal fluorescence microscopy in Section
\ref{appl:deblurring}.

Before we study the aforementioned examples, we clarify some common notation. We
will henceforth assume that $U = H = (\R^m)^d$ with $d=1$ (Section
\ref{appl:reg}) and $d=2$ (Sections \ref{appl:denoising} and
\ref{appl:deblurring}), respectively. Moreover, we will employ gradient based
regularization functionals of the form
\begin{equation}\label{appl:tvpenalty}
  J(u) = TV_p(u) := \frac{1}{p}\sum_{\vec\nu\in X} \abs{\D
  u_{\vec\nu}}_2^p\quad\text{ with }p\in\set{1,2}
\end{equation}
where $\abs{\cdot}_2$ is the Euclidean norm in $\R^d$ and $\D$ denotes the
forward difference operator defined by
\begin{equation*}
  (\D u_{\vec\nu})_i = \begin{cases}
u_{\vec\nu + \vec {e_i}} - u_{\vec\nu} & \text{ if }
1\leq \nu_i\leq n-1 \\
0 & \text{ else.}
\end{cases}
\end{equation*}

For the case $p = 2$ the minimization problem \eqref{ala:primal} amounts to
solve an implicit time step of the $d$-dimensional diffusion equation with
initial value $(Y + \lambda p_{k-1} - v_k)$ and time step size $\lambda$. This can
be solved by a simple (sparse) matrix inversion. 

For the case $p =1$, $TV_1$ is better known as \emph{total-variation semi-norm}.
It was shown in \cite{frick:MamGee97} (see also \cite{frick:Gra07} for similar
results in the continuous setting) that the \emph{taut-string algorithm} (as
introduced in \cite{frick:DK01}) constitutes an efficient solution method for
\eqref{ala:primal} in the case $d = 1$. In the general case $d\geq 1$, we employ
the fixed point approach for solving the Euler-Lagrange equations for
\eqref{ala:primal} described in \cite{frick:DV97} (see also \cite[Chap.
8.2.4]{frick:Vog02}). We finally note that the functional $TV_1$ fails to be
differentiable; a fact that leads to serious numerical problems when trying to
compute the Euler-Lagrange conditions for \eqref{ala:primal}. Hence, we will use
in our simulations a regularized version of $TV_1$ defined by
\begin{equation}\label{appl:tvreg}
  TV_1^\beta(u) = \sum_{\vec\nu\in X} \sqrt{ (\D u_{\vec\nu})_i^2 + \beta^2}
\end{equation}
for a small constant $\beta > 0$.

\subsubsection*{Evaluation.} In order to evaluate the performance of SMREs,
we will employ various distance measures between an estimator $\hat u$ and the
true signal $u^0$. On the one hand, we will use standard measures such as
\emph{mean integrated squared error (MISE)} and the \emph{mean integrated
absolute error (MIAE)} which are given by
\begin{equation*}
  \text{MISE} = \E{\frac{1}{m^d}\sum_{\vec\nu\in X} (\hat u_{\vec\nu} -
  u^0_{\vec\nu})^2}\quad\text{ and }\quad
  \text{MIAE} = \E{\frac{1}{m^d}\sum_{\vec\nu\in X} \abs{\hat u_{\vec\nu} -
  u^0_{\vec\nu}}},
\end{equation*}
respectively. On the other hand, we also intend to measure how well an
estimator $\hat u$ matches the ``smoothness'' of the true signal $u^0$, where
smoothness is characterized by the regularization functional $J$. To this end, we
introduce the \emph{symmetric Bregman divergence}
\begin{equation*}
  D^{\text{sym}}_J(\hat u,u^0) = \frac{1}{m^d}\sum_{\vec\nu\in X} \left(\nabla
  J(\hat u)_{\vec\nu} - \nabla J(u^0)_{\vec\nu}\right)\left(\hat
  u_{\vec\nu}-u^0_{\vec\nu}\right),
\end{equation*}
where $\nabla J$ denotes the gradient of the regularization functional $J$. Clearly,
$D^{\text{sym}}_J(\hat u,u^0)$ is symmetric and since $J$ is assumed to be
convex, also non-negative. However, the symmetric Bregman divergence usually
does not satisfy the triangle inequality and hence in general does not define a
(semi-) metric on $U$ \cite{frick:Csi91}. The following
examples shed some light on how the Bregman divergence incorporates the
functional $J$ in order to measures the distance of $\hat u$ and $u^0$.
\begin{example}\label{appl:exbreg} 
Let $J(u) = TV_p$ as in
\eqref{appl:tvpenalty}. \begin{enumerate}[i)]
    \item If $p=2$ , then
	\begin{equation*}
  	D^{\text{sym}}_J(\hat u,u^0) = \sum_{\vec\nu\in X} \abs{\D \hat u_{\vec\nu} - \D
  	u^0_{\vec\nu}}_2^2. 
    \end{equation*}    
    In other words, the symmetric Bregman distance w.r.t. to $TV_2$ is the mean
    squared distance of the \emph{derivatives} of $\hat u$ and $u^0$. 
    \item If $p=1$, then
	\begin{equation*}
          \begin{split}
  D^{\text{sym}}_J(\hat u, u^0) & = \frac{1}{m^2}\sum_{\vec\nu \in X}
  \left(\frac{\D\hat u_{\vec\nu}}{\abs{\D \hat u_{\vec\nu}}} - \frac{\D
  u^0_{\vec\nu}}{\abs{\D u^0_{\vec\nu}}}\right)\cdot\left(\D
  \hat u_{\vec\nu} - \D u^0_{\vec\nu}\right) \\
  & = \frac{1}{m^2}\sum_{\vec\nu \in X} \left(\abs{\D \hat u_{\vec\nu}}
  + \abs{\D u^0}\right)\left(1-\frac{\D\hat u_{\vec\nu}}{\abs{\D \hat
  u_{\vec\nu}}} \cdot \frac{\D
  u^0_{\vec\nu}}{\abs{\D u^0_{\vec\nu}}}\right) \\ & =
  \frac{1}{m^2}\sum_{\vec\nu \in X} \left(\abs{\D \hat u_{\vec\nu}} + \abs{\D u^0}\right)(1-\cos \gamma_{\vec\nu}),
          \end{split}
    \end{equation*}   
    where $\gamma_{\vec\nu}$ denotes the angle between the level lines of $\hat
    u$ and $u^0$ at the point $x_{\vec\nu}$. Put differently, the symmetric
    Bregman divergence w.r.t the total variation semi-norm $TV_1$ is small if
    for sufficiently many points $x_{\vec\nu}$ either both $\hat u$ and $u^0$
    are constant in a neighborhood of $x_{\vec\nu}$ or the level lines of
    $\hat u$ and $u^0$ through $x_{\vec\nu}$ are parallel. In practice rather $TV_1^\beta$ in \eqref{appl:tvreg} (for a small $\beta >
    0$) instead of $TV_1$ is used  in order to avoid singularities. Then, the above formulas are
    slightly more complicated. 
  \end{enumerate}
\end{example}

We will use the mean symmetric Bregman divergence (MSB) given by
\begin{equation*}
  \text{MSB} = \E{D^{\text{sym}}_J(\hat u,u^0)}
\end{equation*}
as an additional evaluation method. In all our simulations we approximate the
expectations above by the empirical means of $500$ trials. 

\subsubsection*{Comparison with other methods.} We will compare the SMREs to
other regression methods. Firstly, we will consider estimators obtained by the
\emph{global} penalized least squares method:
\begin{equation}\label{appl:rof}
  \hat u(\lambda) := \argmin_{u\in H} \frac{1}{2} \sum_{\vec\nu\in
  X}(u_{\vec\nu} - Y)^2 + \lambda J(u),\quad \lambda > 0.
\end{equation}
In particular, we focus on estimators $\hat u (\lambda)$ that are closest (in
some sense) to the true function $u^0$. We call such estimators \emph{oracles}.
We define the $\Ls{2}$- and Bregman-oracle by $\hat u_{\Ls{2}} = \hat
u(\lambda_2)$ and $\hat u_{\text{B}} = \hat u(\lambda_{\text{B}})$, where
\begin{equation*}
  \lambda_2 := \E{\argmin_{\lambda > 0} \norm{ u^0 - \hat u(\lambda)}}
  \quad\text{ and }\quad  \lambda_{\text{B}} := \E{\argmin_{\lambda > 0}
  D_J^{\text{sym}}(u^0,\hat u(\lambda))} 
\end{equation*}
respectively. Of course, oracles are not available in practice, since the
true signal $u^0$ is unknown. However, they represent ideal instances within
the class of estimators given by \eqref{appl:rof} that usually perform better than any
data-driven parameter choices (such as cross-validation) and hence may serve as
a reference.

Secondly, we also compare our approach to \emph{adaptive weights smoothing
(AWS)} \cite{frick:PolSpo00} which constitutes a benchmark technique for data-driven,
spatially adaptive regression. We compute these estimators by means of the
official R-package\footnote{available
at\url{http://cran.r-project.org/web/packages/aws/index.html}} and denote them by
$\hat u_{\text{aws}}^{\ker}$, where $\ker\in\set{\text{Gaussian},
\text{Triangle}}$ decodes the shape of the underlying regression kernel.

\subsection{Non-parametric Regression}\label{appl:reg}

In this section we apply the SMRE technique to a nonparametric 
regression problem in $d=1$ dimensions, i.e. the noise model
\eqref{intro:lineqn} becomes
\begin{equation}\label{appl:regr}
  Y_{\vec\nu} = u^0_{\vec\nu} + \eps_{\vec\nu}\quad \vec\nu=1,\ldots,m,
\end{equation}
where we assume that $\eps_{\vec\nu}$ are independently and normally distributed
r.v. with $\E{\eps_\nu} = 0$ and $\E{\eps_\nu^2}=\sigma^2$. The upper left
image in Figure \ref{appl:fig:one} depicts the true signal $u^0$ (solid line)
and the data $Y$, with $m=1024$ and $\sigma = 0.5$. The application we have in
mind with this example arises in NMR spectroscopy, where the NMR spectra
provide structural information on the number and type of chemical entities in a
molecule. In this context, we suggest to choose $J = TV_2$, since the true
signal $u^0$ is rather smooth (see \cite{frick:DavKovMei09} for examples where
$J$ is chosen to be the total variation of the first and second derivative).

Finally, we discuss the MR-statistic $T$ in \eqref{intro:mrstateqn}. We choose
$\Lambda = \id$ and the index set $\S$ to consist of all discrete intervals
with side lengths ranging from $1$ to $100$. For an interval $S\in \S$ we set
$\omega^S=(\#S)^{-1\slash 2}\chi_S$. Thus, each SMRE solves the
constrained optimization problem
\begin{equation}\label{appl:regrsmre}
  \inf_{u\in U} TV_2(u)\quad\text{ s.t. }\quad \frac{1}{\sqrt{\#
  S}}\abs{\sum_{\vec\nu\in S} (Y-u)_{\vec\nu}} \leq q\quad\forall(S\in \S).
\end{equation}
We choose $q$ to be the $\alpha$-quantile of the MR-statistic $T$, that is
\begin{equation}\label{appl:quantile} 
  q_\alpha = \inf\set{q\in\R ~:~ \Prob\left(\max_{S\in\S}\frac{1}{\sqrt{\#
  S}}\abs{\sum_{\vec\nu\in S} \eps_{\vec\nu}}\leq q\right)\geq \alpha }\quad
  \alpha\in (0,1).
\end{equation}
We note that except for few special cases (cf.
\cite{frick:HotMarStiDavKabMun12, frick:Kab10}) closed form expressions for the
distribution of the MT-statistic $T$ are usually not at hand. In practice one
rather considers the empirical distribution of $T$ where the variance
$\sigma^2$ can be estimated at a rate $\sqrt{md}$ (cf.
\cite{frick:MunBisWagFre05}).  

We will henceforth denote by $\hat u_\alpha$ a solution of \eqref{appl:regrsmre}
with $q = q_\alpha$. As argued in Section \ref{intro:smre}, $\hat u_\alpha$ is
smoother (i.e. has smaller value $TV_2$) than the true signal $u^0$ with a
probability of at least $\alpha$ while it satisfies the constraint that the
multiresolution statistic $T$ does not exceed $q_\alpha$.  This is a sound
statistical interpretation of the regularization parameter $\alpha$.
 
\begin{figure}[h]
\begin{center}
\includegraphics[width = 0.3\imwidth]{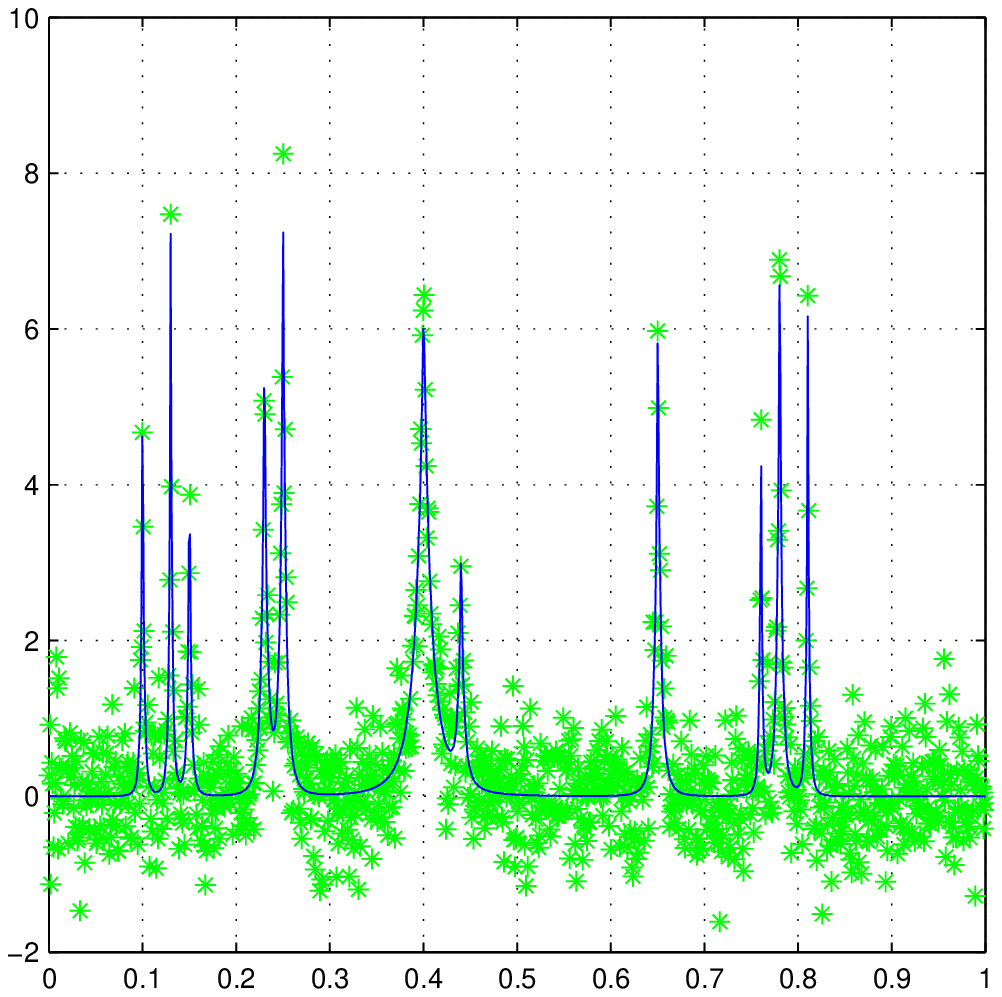}
\hspace{0.01\imwidth}
\includegraphics[width = 0.3\imwidth]{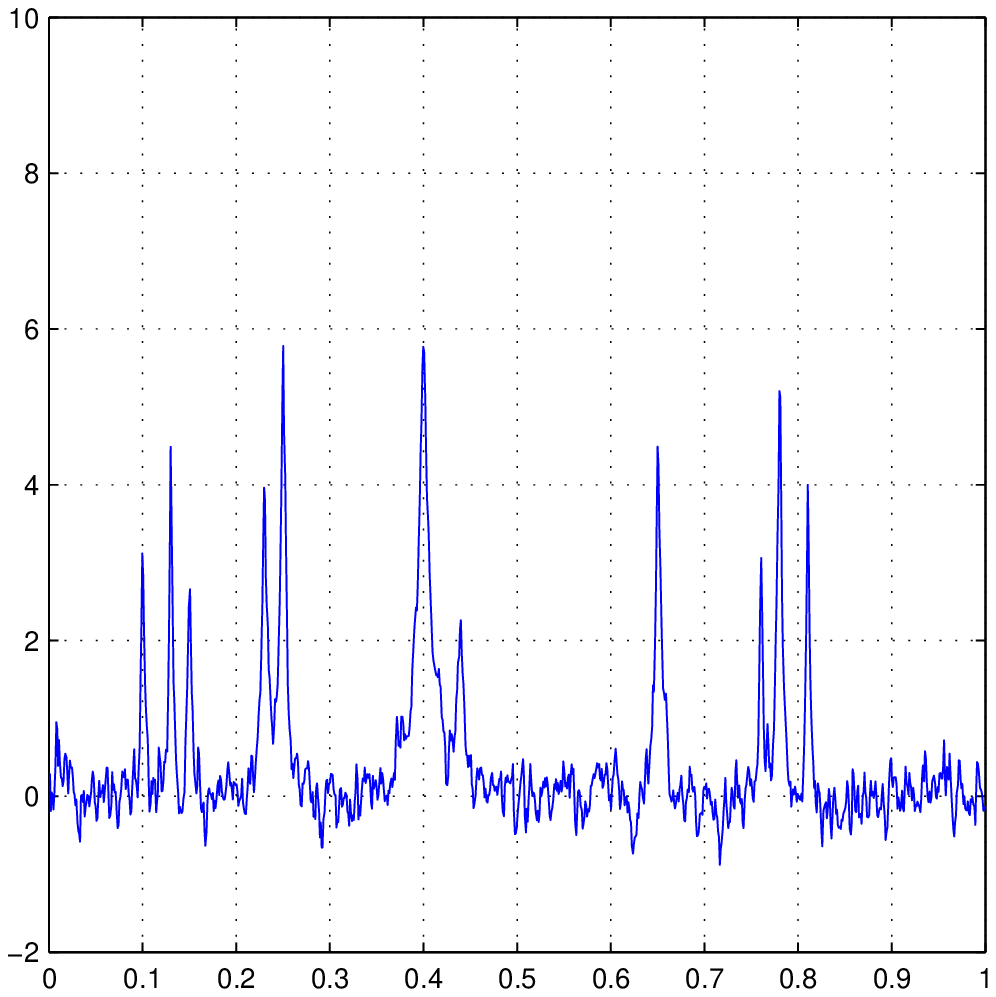}
\hspace{0.01\imwidth}
\includegraphics[width = 0.3\imwidth]{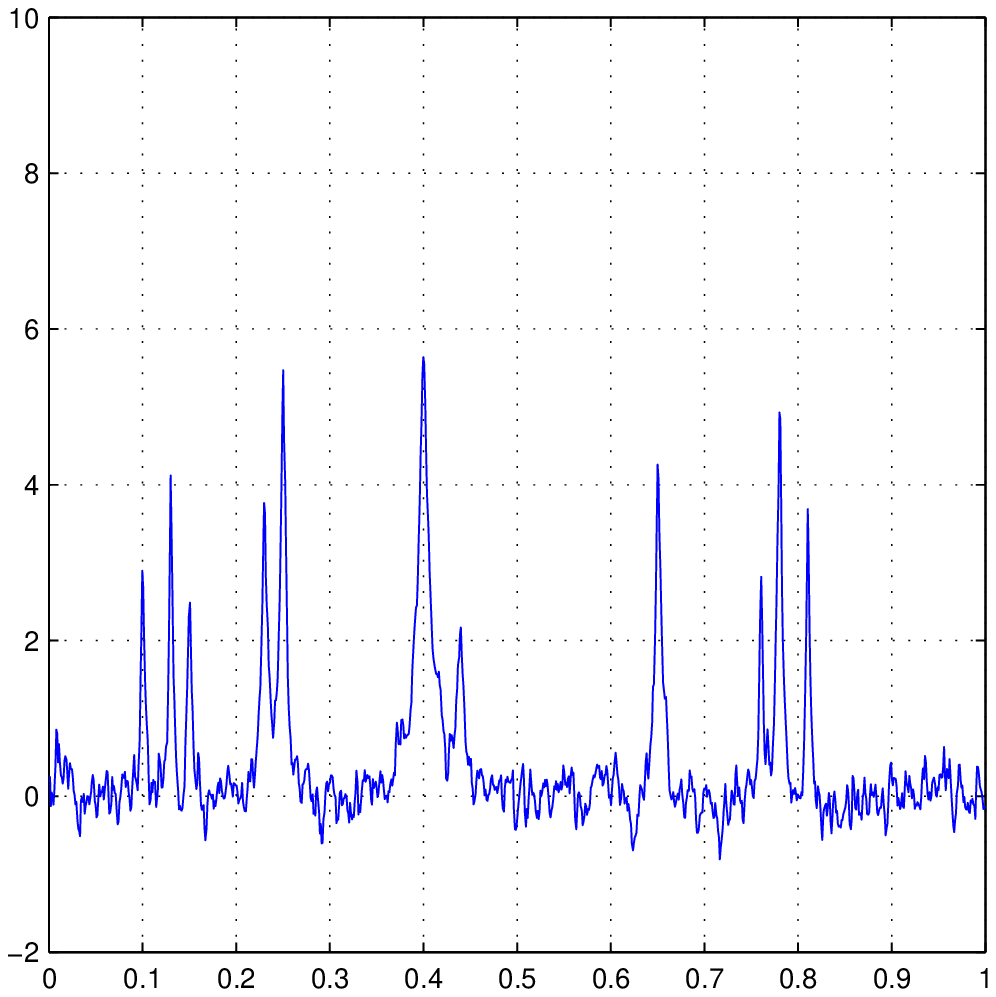}

\vspace{0.01\imwidth}

\includegraphics[width = 0.3\imwidth]{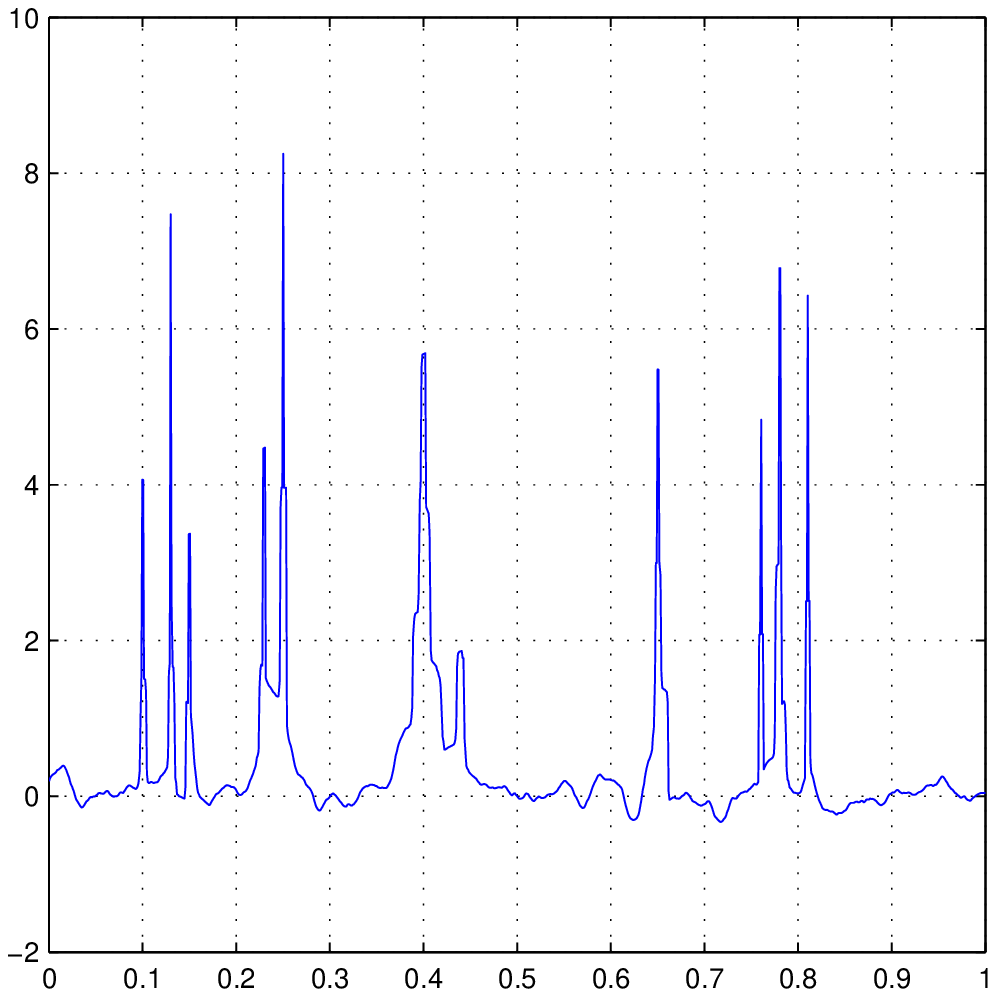}
\hspace{0.01\imwidth}
\includegraphics[width = 0.3\imwidth]{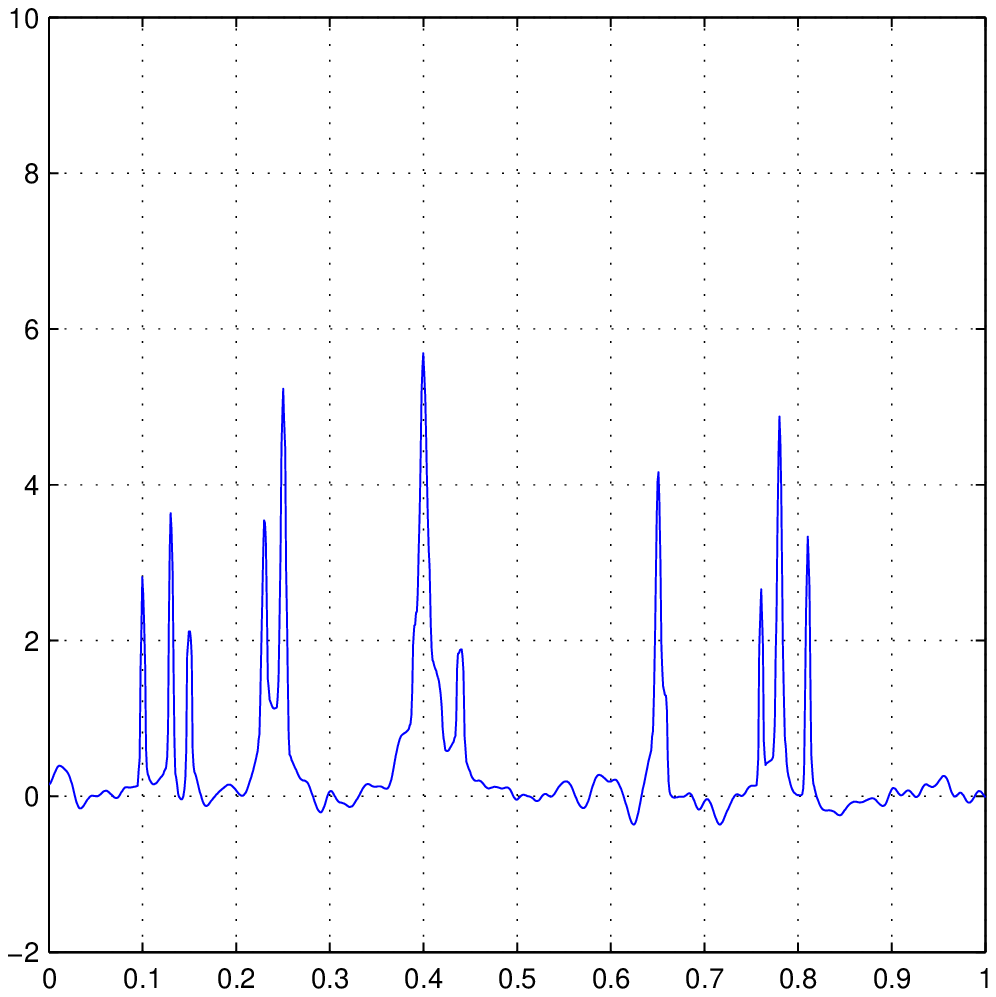}

\vspace{0.01\imwidth}

\includegraphics[width = 0.3\imwidth]{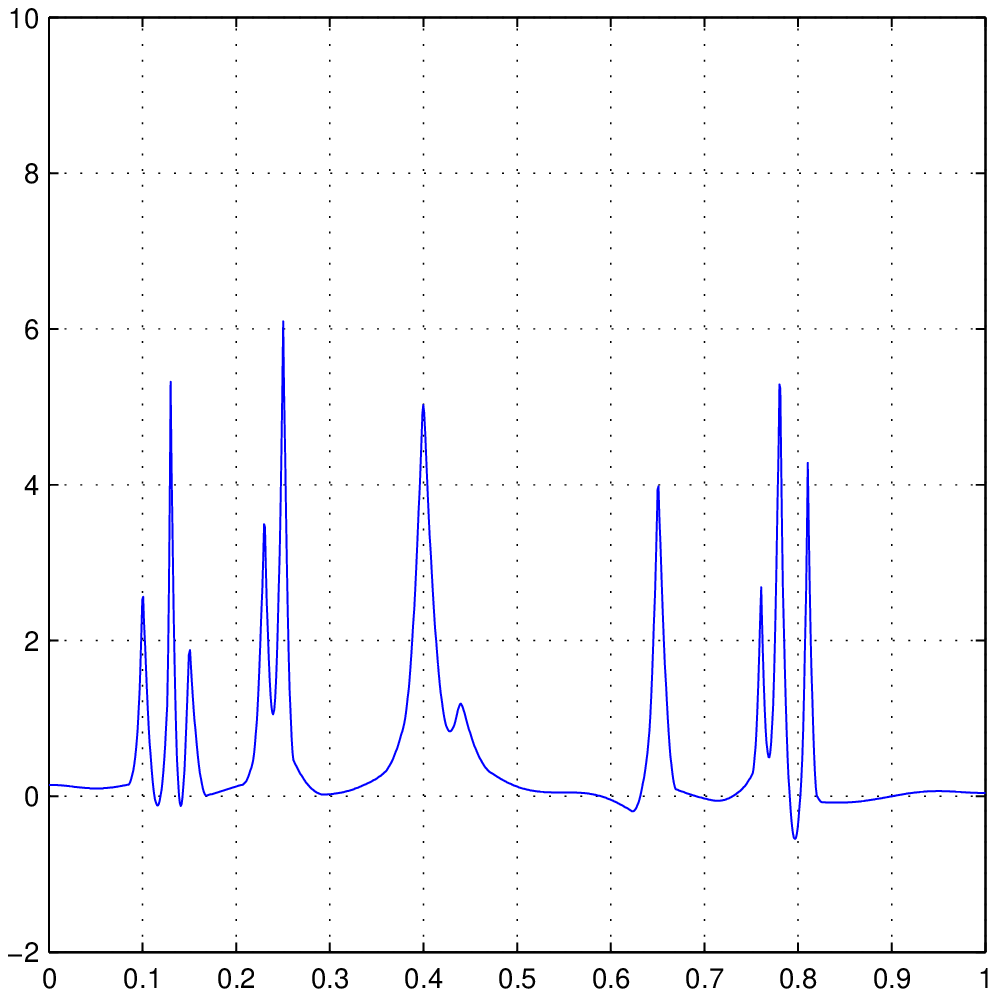}
\hspace{0.01\imwidth}
\includegraphics[width = 0.3\imwidth]{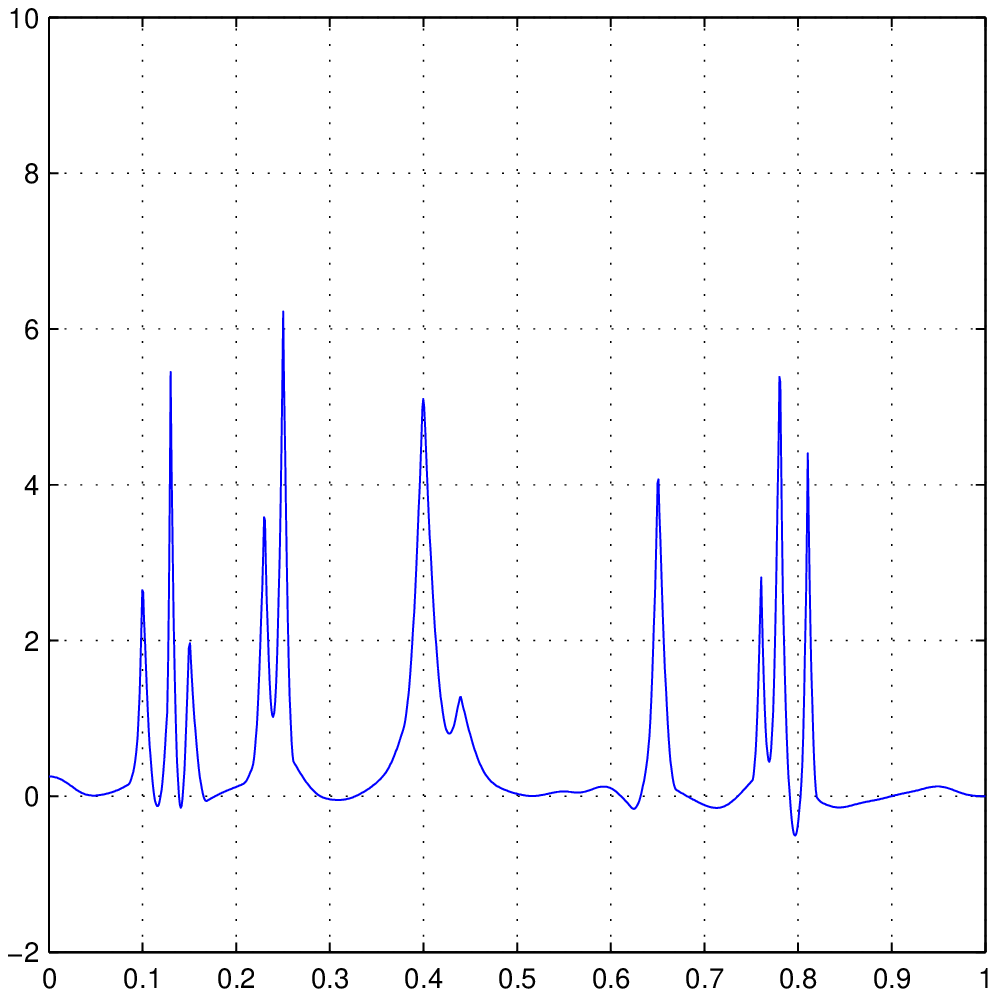}
\hspace{0.01\imwidth}
\includegraphics[width = 0.3\imwidth]{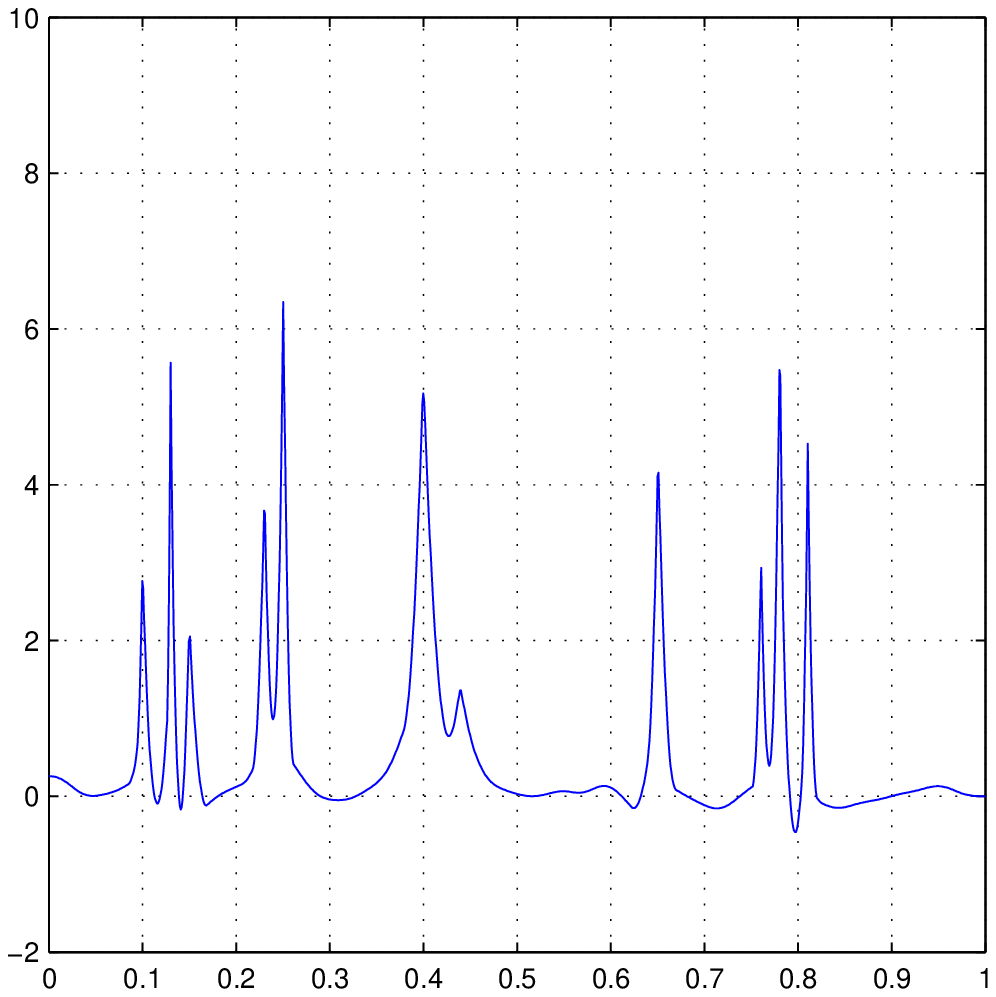}
\end{center}
\caption{Row-wise from top left: true signal $u^0$ (solid) with data $Y$, oracles
$\hat u_{\Ls{2}}$ and $\hat u_{\text{B}}$, AWS estimators $\hat
u_{\text{aws}}^{\text{Triangle}}$ and $\hat u_{\text{aws}}^{\text{Gaussian}}$,
and the SMREs $\hat u_{0.9}$, $\hat u_{0.75}$ and $\hat
u_{0.5}$.}\label{appl:fig:one}
\end{figure}
 
\subsubsection*{Numerical results and simulations.} In Figure \ref{appl:fig:one} the
oracles $\hat u_{\Ls{2}}$ and $\hat u_{\text{B}}$, the AWS-estimators $\hat u_{\text{aws}}^{\text{Triangle}}$ and and
$\hat u_{\text{aws}}^{\text{Gauss}}$ as well as the SMRE $\hat u_{0.9}$, $\hat
u_{0.75}$ and $\hat u_{0.5}$ are depicted. It is evident that the SMRE matches
the smoothness of the true object much better than the other estimators while
the essential features of the signal (such as peak location and
peak height) are preserved. In particular, almost no additional local extrema
are generated by our approach which stays in obvious contrast to the other
methods.  Moreover, we point out that the SMRE are quite  robust w.r.t.\ the
choice of the  confidence level $\alpha$.

We verify this behavior by a simulation study in Table \ref{appl:sim1}. For
different noise levels ($\sigma = 0.1, 0.3$ and $0,5$) we compare the MISE, MIAE
and MSB. Additionally, we compute the \emph{mean number of local maxima (MLM)} of
$\hat u$ relative to the number of local maxima in $u^0$ (which is $11$). Here
$\hat u$ is any of the above estimators. Note that the latter measure (similar to
the MSB) takes into account the smoothness of the estimators where a value
$\text{MLM}\gg 1$ indicates too many local maxima and hence a lack of regularity
whereas $\text{MLM} < 1$ implies severe oversmoothing.

 \begin{table}[h!]
%\begin{center}
{\scriptsize
\begin{tabular}{|l|c|c|c|c||c|c|c|c|}
\hline
 & \multicolumn{4}{|c||}{$\sigma = 0.1$} & \multicolumn{4}{|c|}{$\sigma =
 0.3$} \\ 
 \hline\hline
 &  MISE  & MSB & MIAE & MLM & MISE  & MSB & MIAE & MLM   \\
 \hline
 $\hat u_{\Ls{2}}$ & 0.009 & 0.008 & 0.071 & 11.881 & 0.046 & 0.027 &
 0.156 & 10.915  \\
 $\hat u_{\text{B}}$ & 0.009 &  0.007 &  0.070 & 11.700 & 0.048 & 
 0.026 & 0.149 & 10.359 \\
 $\hat u_{\text{aws}}^{\text{Triangle}}$ & 0.007 & 0.007 & 0.048 & 2.551 &
 0.040 & 0.035 & 0.112 & 3.053  \\
 $\hat u_{\text{aws}}^{\text{Gauss}}$ & 0.054 & 0.040 & 0.068 & 1.971 &
 0.062 & 0.041 & 0.107 & 2.230  \\
 \hline
 $\hat u_{0.9}$ & 0.008 & 0.004 & 0.047 & 1.336 & 0.056 & 0.019 &
 0.127 & 1.273  \\
 $\hat u_{0.75}$ & 0.007 & 0.004 & 0.044 & 1.342 & 0.050 & 0.018 &
  0.121 & 1.290   \\
 $\hat u_{0.5}$ & 0.007 & 0.004 &  0.043 & 1.366 & 0.046 & 0.017 &
  0.116 & 1.290    \\
 \hline
\end{tabular}

\vspace{0.01\imwidth}
\begin{tabular}{|l|c|c|c|c|}
\hline
  & \multicolumn{4}{|c|}{$\sigma = 0.5$}\\ 
 \hline\hline
 & MISE  & MSB & MIAE & MLM  \\
 \hline	
 $\hat u_{\Ls{2}}$ & 0.091 & 0.037 & 0.213 & 9.860 \\
 $\hat u_{\text{B}}$ & 0.094 & 0.036 & 0.206 & 9.135 \\
 $\hat u_{\text{aws}}^{\text{Triangle}}$ & 0.078 & 0.058 & 0.162 & 3.141 \\
 $\hat u_{\text{aws}}^{\text{Gauss}}$ & 0.079 & 0.043 & 0.149 & 2.330 \\
 \hline
 $\hat u_{0.9}$  & 0.134 & 0.034 & 0.207  & 1.194 \\
 $\hat u_{0.75}$ & 0.120 & 0.032 & 0.196 & 1.241  \\
 $\hat u_{0.5}$  & 0.109 & 0.030 & 0.186  & 1.238  \\
 \hline
\end{tabular}
\caption{Simulation studies for one dimensional peak data
set.}\label{appl:sim1}}
%\end{center}
\end{table}

As it becomes apparent from Table \ref{appl:sim1}, the SMREs are
performing similarly well when compared to the reference estimators as far as
the standard measures MISE and MIAE are concerned. For small noise levels ($\sigma
= 0.1$) SMREs even prove to be superior. The distance measures MSB
and MLM, however, are significantly smaller for SMREs which
indicates that these meet the smoothness of the true object $u^0$ much better
than the reference estimators (cf. Example \ref{appl:exbreg} i)). All in all,
the simulation results confirm our visual impressions above.

\subsubsection*{Implementation Details.} 

The current index set $\S$ results in an overall number of 
constraints in \eqref{appl:regrsmre} of
\begin{equation*}
  \# \S = \sum_{i=1}^{100} (1024 - i + 1) = 97450.
\end{equation*} 
As pointed out in Section \ref{impl:proj}, the efficiency of Dykstra's Algorithm
can be increased by grouping independent side-conditions, that is side-conditions
corresponding to intervals in $\S$ with empty intersection. For example, the
system $\mathcal{S}$ can be grouped such that the intersection of the 
corresponding sets $D_1,\ldots,D_M$ in \eqref{impl:regroup} form
$\mathcal{C}$  with
\begin{equation*}
  M = \sum_{i=1}^{100} i = 5050.
\end{equation*}
In all our simulations we set $\tau = 10^{-4}$ and $\lambda = 1.0$ in Algorithm
\ref{impl:ala} which results in $k[\tau] \approx 100$ iterations and an overall
computation time of approximately $20$ minutes for each SMRE. We note, however,
that more than $95\%$ of the computation time is needed for the projection step
\eqref{ala:noise} and that a considerable speed up for the latter could be
achieved by parallelization.

\subsection{Image denoising}\label{appl:denoising}

In this section we apply the SMRE technique to the problem of image denoising,
that is non-parametric regression in $d=2$ dimensions. In other words, we consider
the noise model \eqref{appl:regr} as in Section \ref{appl:reg}, where
the index $\vec\nu$ ranges over the discrete square $\set{1,\ldots,m}^2$. In
Figure \ref{appl:test_images} two typical examples for images $u^0$ and noisy
observations $Y$ are depicted ($m=512$ and $\sigma=0.1$, where $u^0$ is scaled
between $0$ (black) and $1$ (white)). 

\begin{figure}[h!]
\begin{center}
\includegraphics[width = 0.4\imwidth]{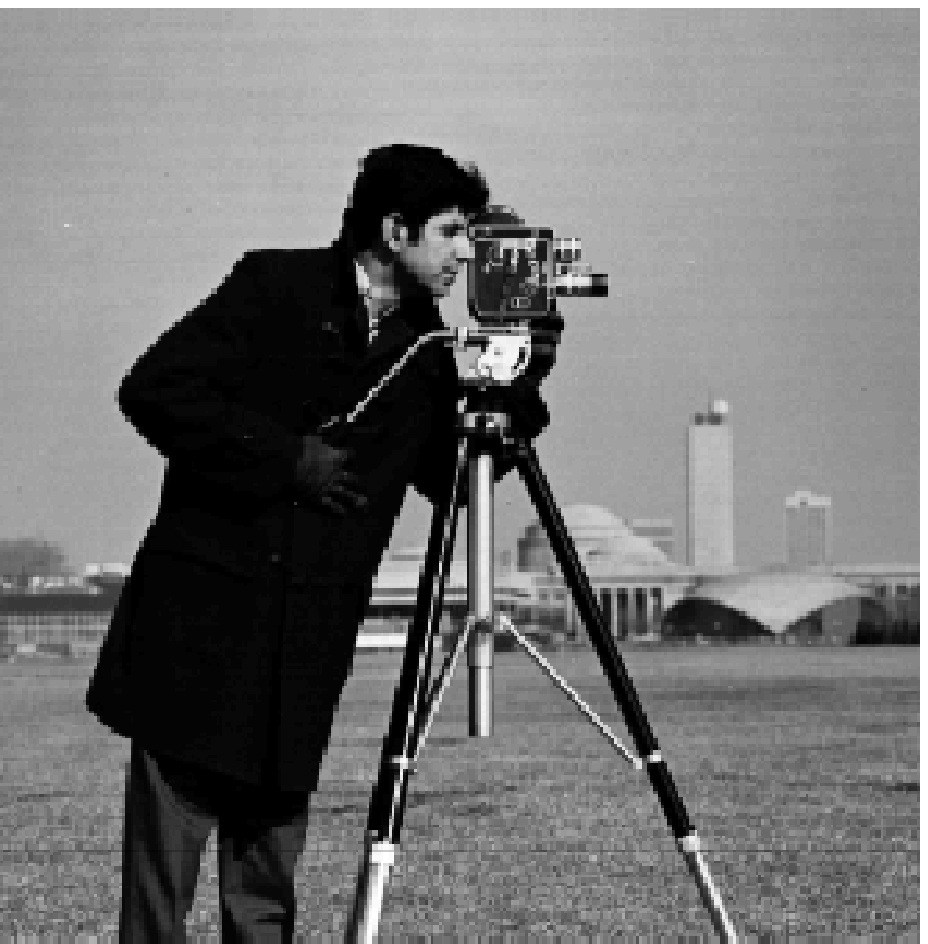}
\hspace{0.1cm}
\includegraphics[width = 0.4\imwidth]{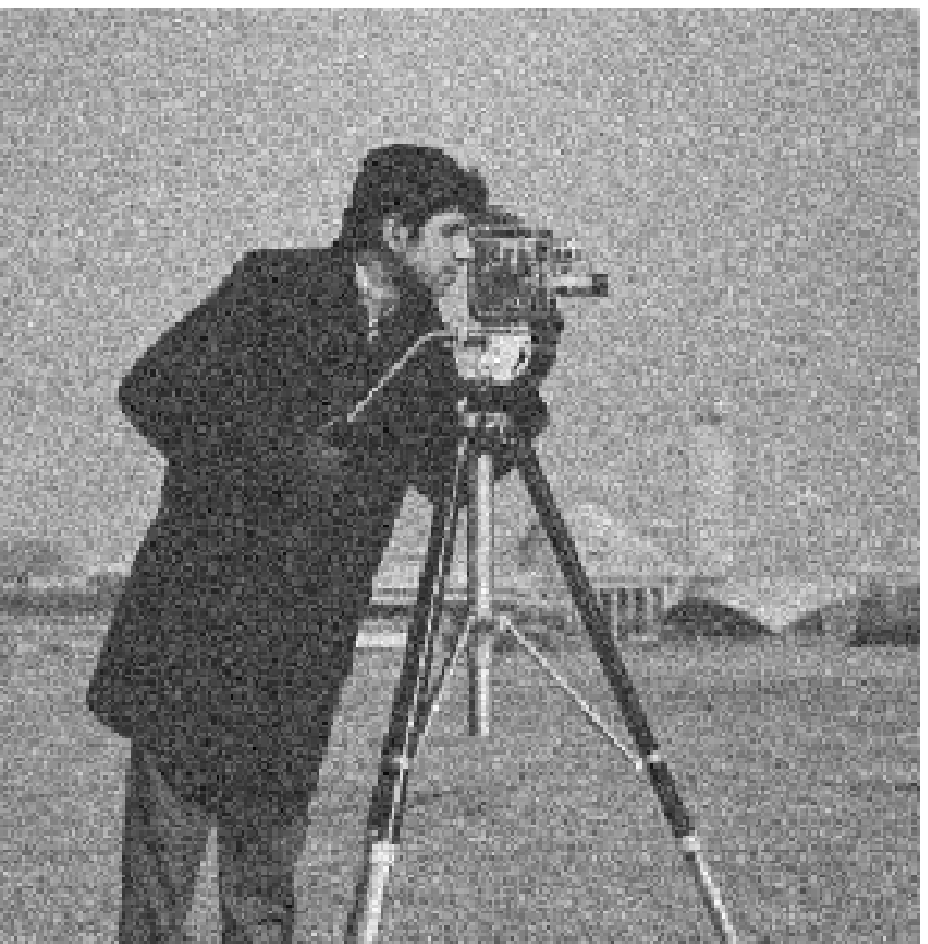}

\vspace{0.1cm}
\includegraphics[width = 0.4\imwidth]{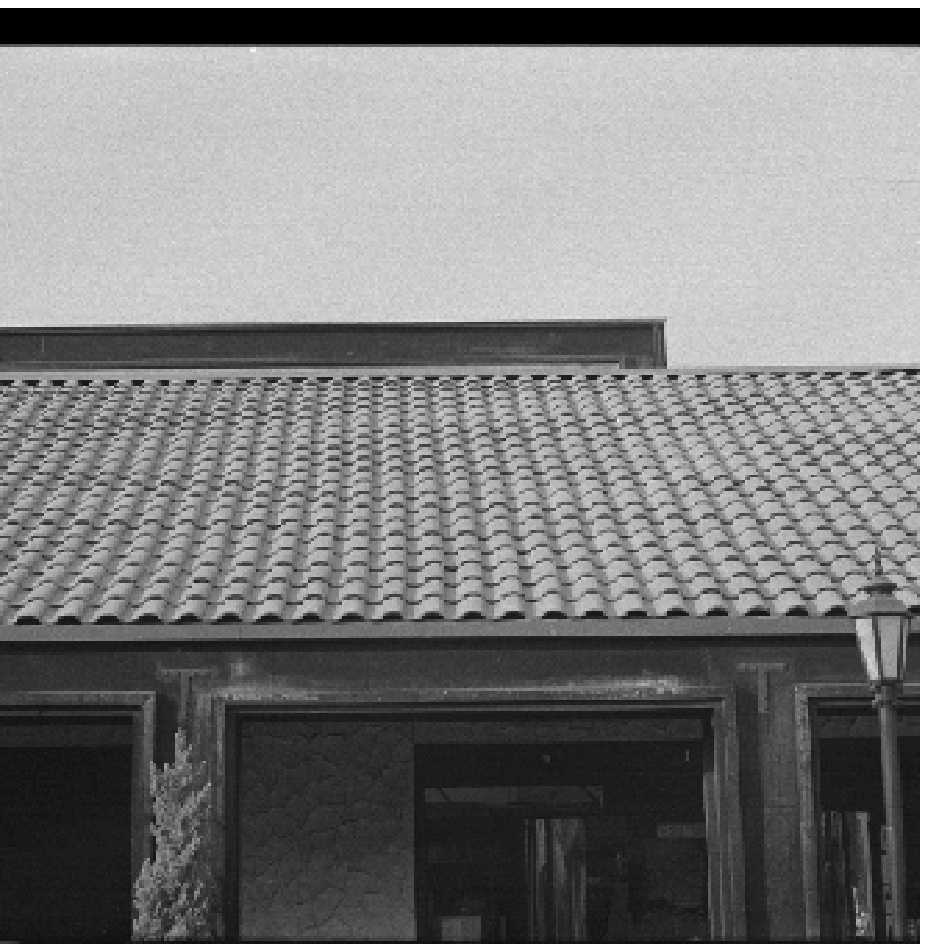}
\hspace{0.1cm}
\includegraphics[width = 0.4\imwidth]{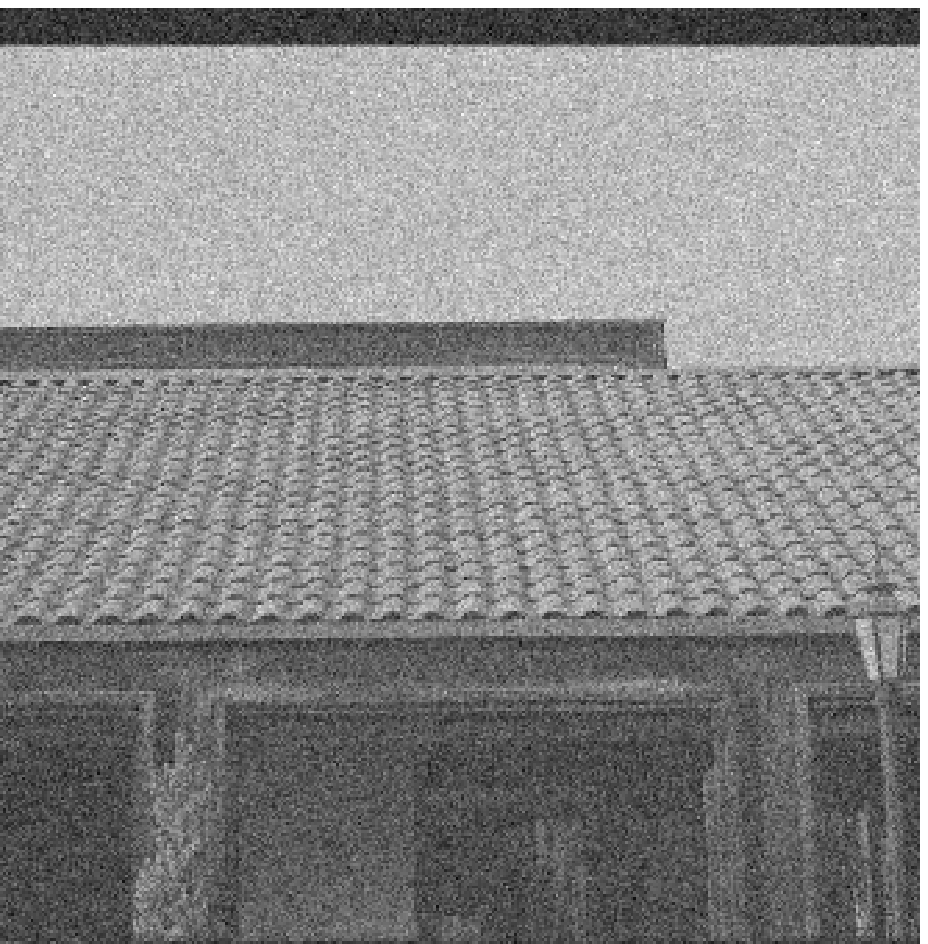}
\caption{Standard test images ``cameraman'' (top)
 and ``roof'' (bottom) and their noisy counterparts.}\label{appl:test_images}
\end{center}
\end{figure}

We will use the total-variation semi-norm $J=TV_1^\beta$ as regularization functional
($\beta = 10^{-8}$). Moreover, we choose $\Lambda$ to be defined as
\begin{equation}\label{appl:lambdasq}
  \Lambda(v)_{\vec\nu} = v_{\vec\nu}^2,\quad \forall(\vec\nu\in {1,\ldots,m}^2).
\end{equation}
The index set $\S$ is defined to be the collection of all discrete squares with
side lengths ranging from $1$ to $25$ and we set $\omega^S = c_S\chi_S$
with yet to be defined constants $c_S$. Thus, each SMREs solves the
constrained optimization problem
\begin{equation}\label{appl:regrsmre2d}
  \inf_{u\in U} TV_1^\beta(u) \quad\text{ s.t. }\quad \sum_{\vec\nu\in S}
  c_S (Y-u)_{\vec\nu}^2 \leq q \quad\forall(S\in \S).
\end{equation}
We agree upon $q = 1$ and specify the constants $c_S$. To this end, compute for
$s=1,\ldots,25$ the quantile values 
\begin{equation*}
  q_{\alpha,s} = \inf\set{q\in\R ~:~ 
  \Prob\left(\max_{\substack{S\in\S \\ \#S = s}}\sum_{\vec\nu\in S}
  \eps_{\vec\nu}^2\leq q\right)\geq 1-\alpha }\quad \alpha\in (0,1)
\end{equation*}
and set $c_S = q_{\alpha,\# S}^{-1}$. In other words, the definition of $c_S$
implies that the true signal $u^0$ satisfies the constraints  in
\eqref{appl:regrsmre2d} \emph{for squares of a fixed side length $s$} with
probability at least $\alpha$. We will henceforth denote by $\hat u_\alpha$ a
solution of \eqref{appl:regrsmre2d}. We remark on this particular choice
of the parameters $\omega_S$ below. 

\subsubsection*{Numerical results and simulations.} In Figures
\ref{appl:resultscamera} and \ref{appl:resultsroof} the oracles $\hat u_{\Ls{2}}$
and $\hat u_{\text{B}}$, the AWS-estimators $\hat
u_{\text{aws}}^{\text{Triangle}}$ and $\hat u_{\text{aws}}^{\text{Gauss}}$ as
well as the SMRE $\hat u_{0.9}$  are depicted (for the ``cameraman'' and ``roof''
test image respectively). It is rather obvious that the $\Ls{2}$-oracles are not
favorable: although relevant details in the image are preserved, smooth parts
(as e.g. the sky) still contain random structures. In contrast, the
estimator $\hat u_{\text{aws}}^{\text{Gauss}}$ preserves smooth areas but looses
essential details. The aws-estimator with triangular kernel performs much better,
however, it gives piecewise constant reconstructions of smoothly varying portions
of the image, which is clearly undesirable. The SMRE and the
Bregman-oracle visually perform superior to the other methods. The good
performance of the Bregman-oracle indicates that the symmetric Bregman distance
is a good measure for comparing images. In contrast to the Bregman-oracle, the
SMRE adapts the amount of smoothing to the underlying image structure: constant
image areas are smoothed nicely (e.g. sky portions), while oscillating patterns
(e.g. the grass part in the ``cameraman'' image or the roof tiles in the ``roof''
image) are recovered.

\begin{figure}[h!]
\begin{center}
\includegraphics[width = 0.4\imwidth]{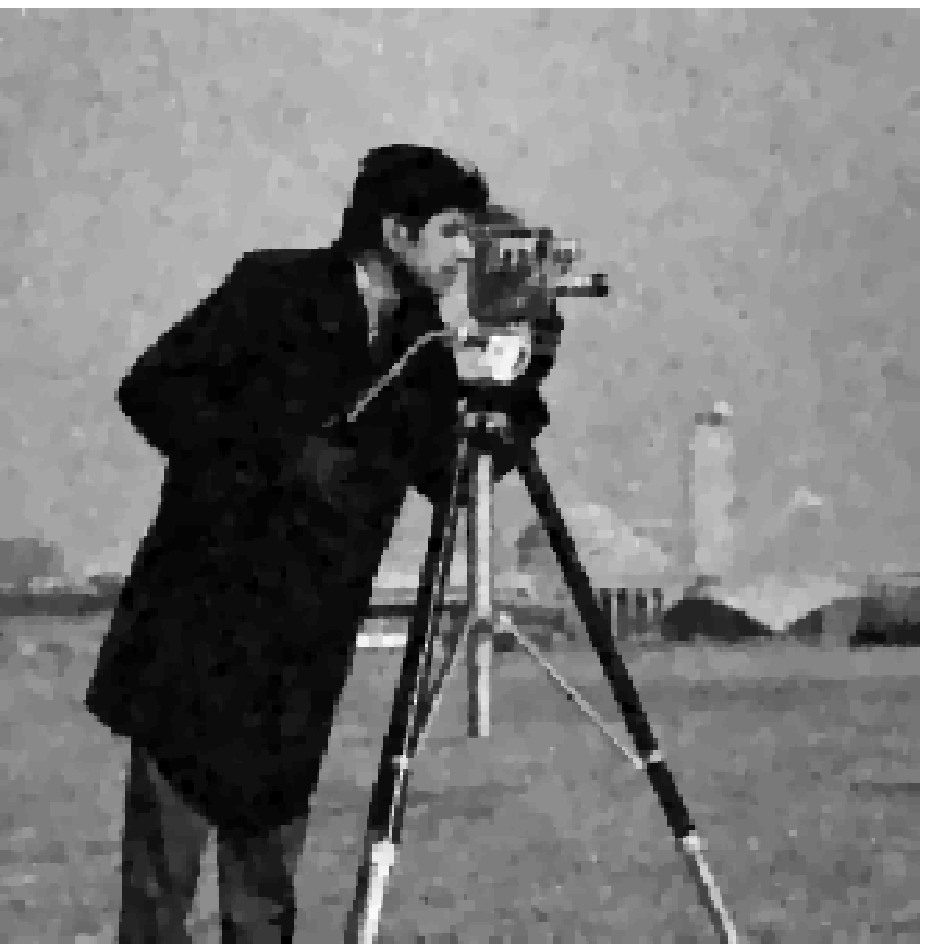}
\hspace{0.01\imwidth}
\includegraphics[width = 0.4\imwidth]{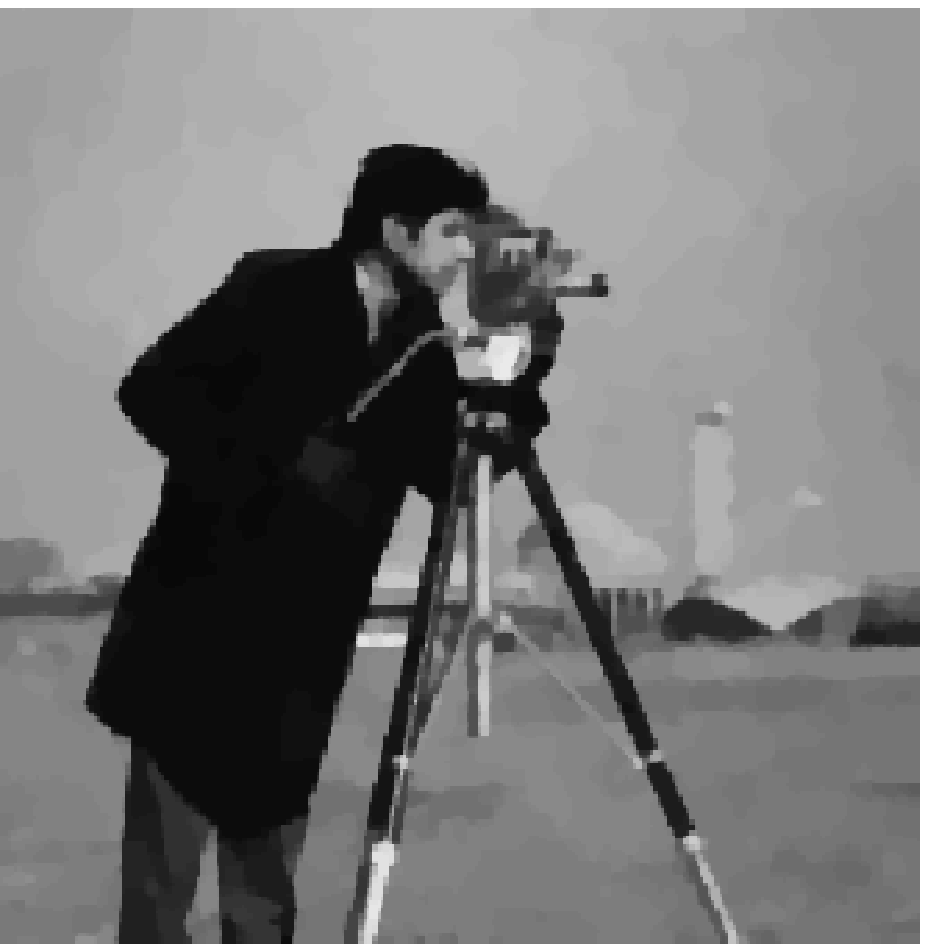}

\vspace{0.01\imwidth}

\includegraphics[width = 0.4\imwidth]{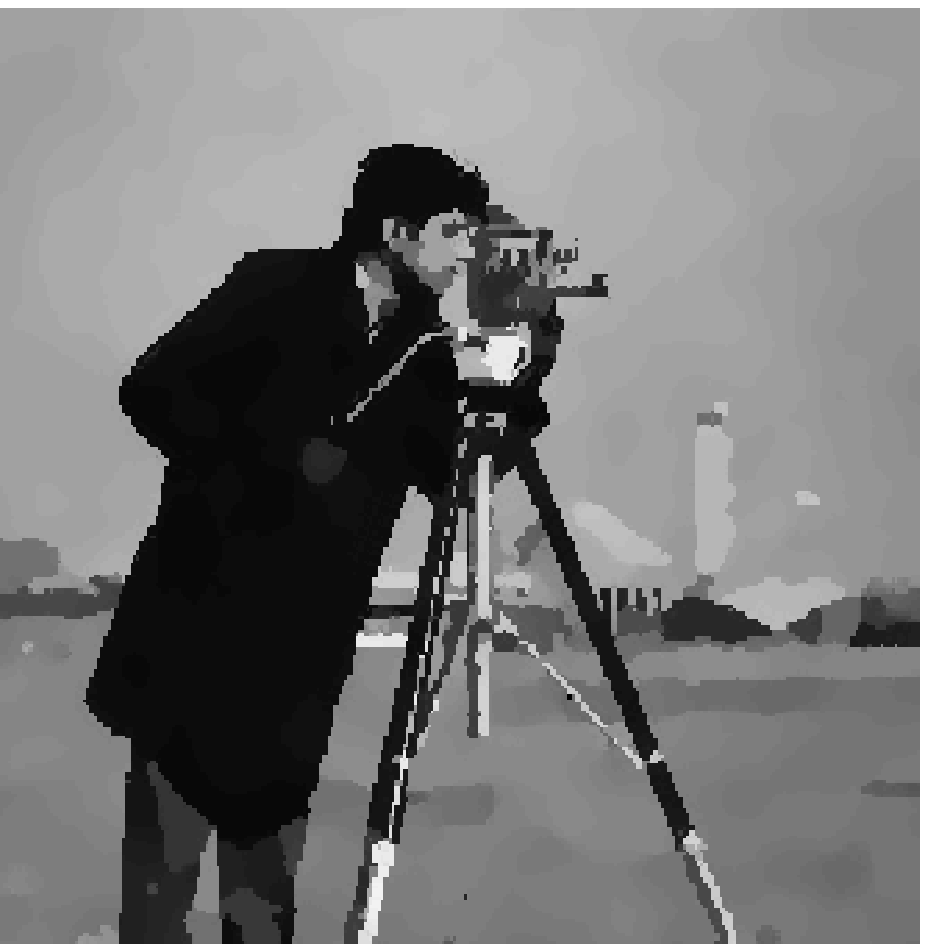}
\hspace{0.01\imwidth}
\includegraphics[width = 0.4\imwidth]{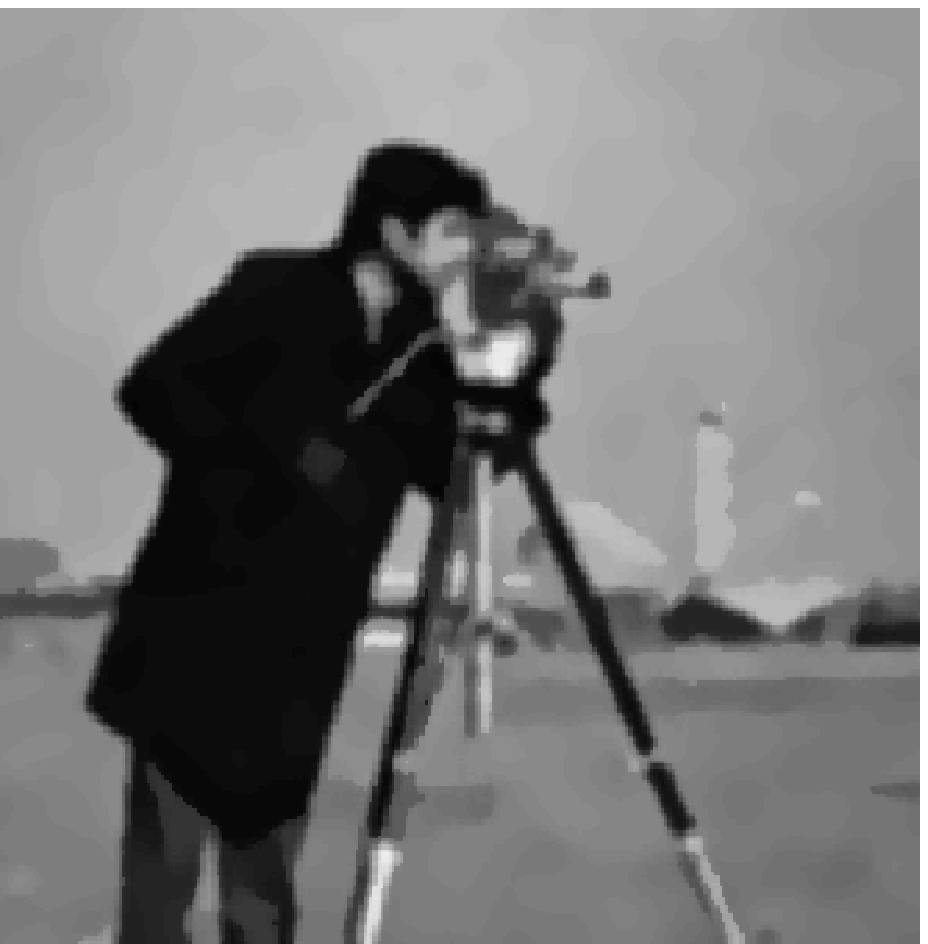}

\vspace{0.01\imwidth}

\includegraphics[width =
0.4\imwidth]{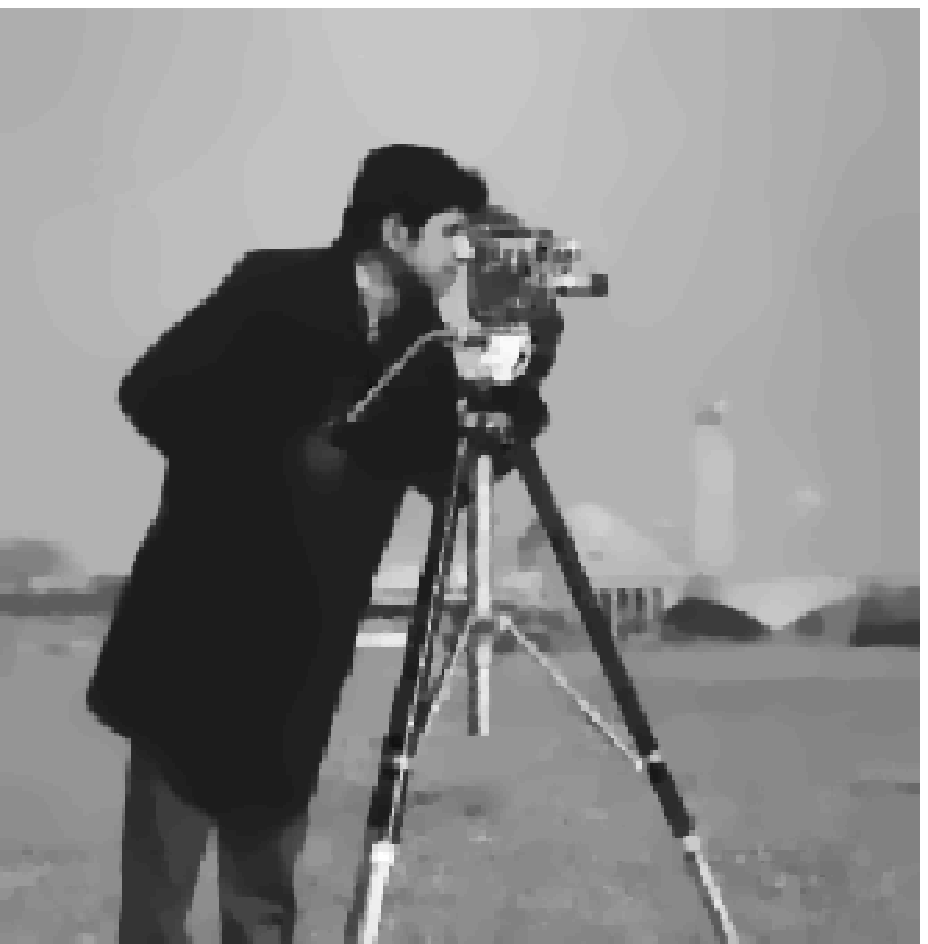}
\caption{Reconstructions ``cameraman'' (row-wise from top left): $\Ls{2}$-oracle
$\hat u_{\Ls{2}}$, Bregman-oracle $\hat u_B$, AWS estimators $\hat u_{\text{aws}}^{\text{Triangle}}$ and $\hat u_{\text{aws}}^{\text{Gaussian}}$,
and SMRE $\hat u_{0.9}$.}\label{appl:resultscamera}
\end{center}
\end{figure}

\begin{figure}[h!]
\begin{center}
\includegraphics[width = 0.4\imwidth]{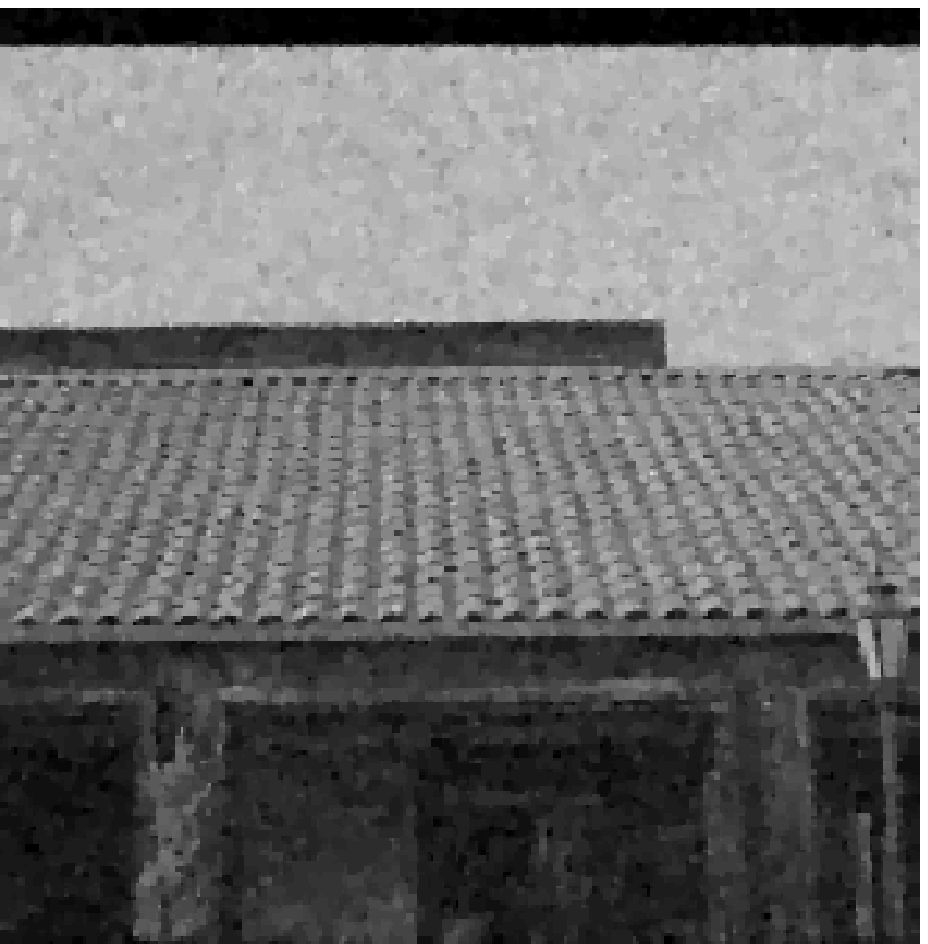}
\hspace{0.01\imwidth}
\includegraphics[width = 0.4\imwidth]{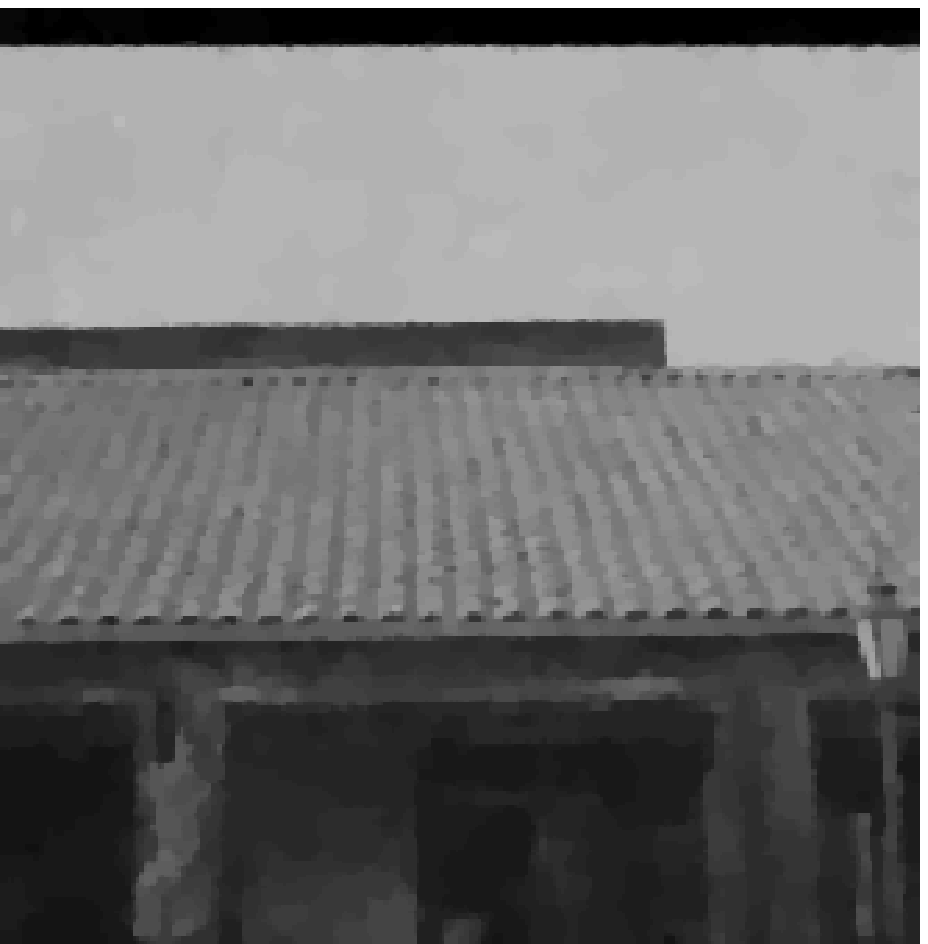}

\vspace{0.01\imwidth}

\includegraphics[width = 0.4\imwidth]{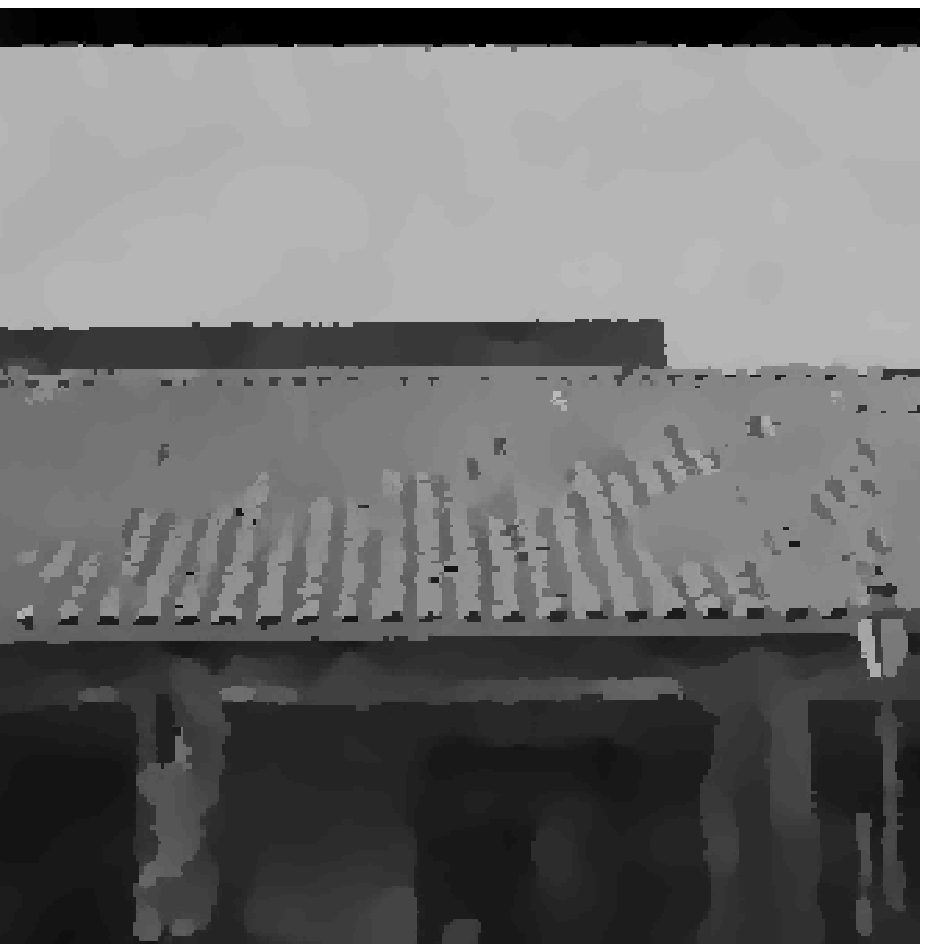}
\hspace{0.01\imwidth}
\includegraphics[width = 0.4\imwidth]{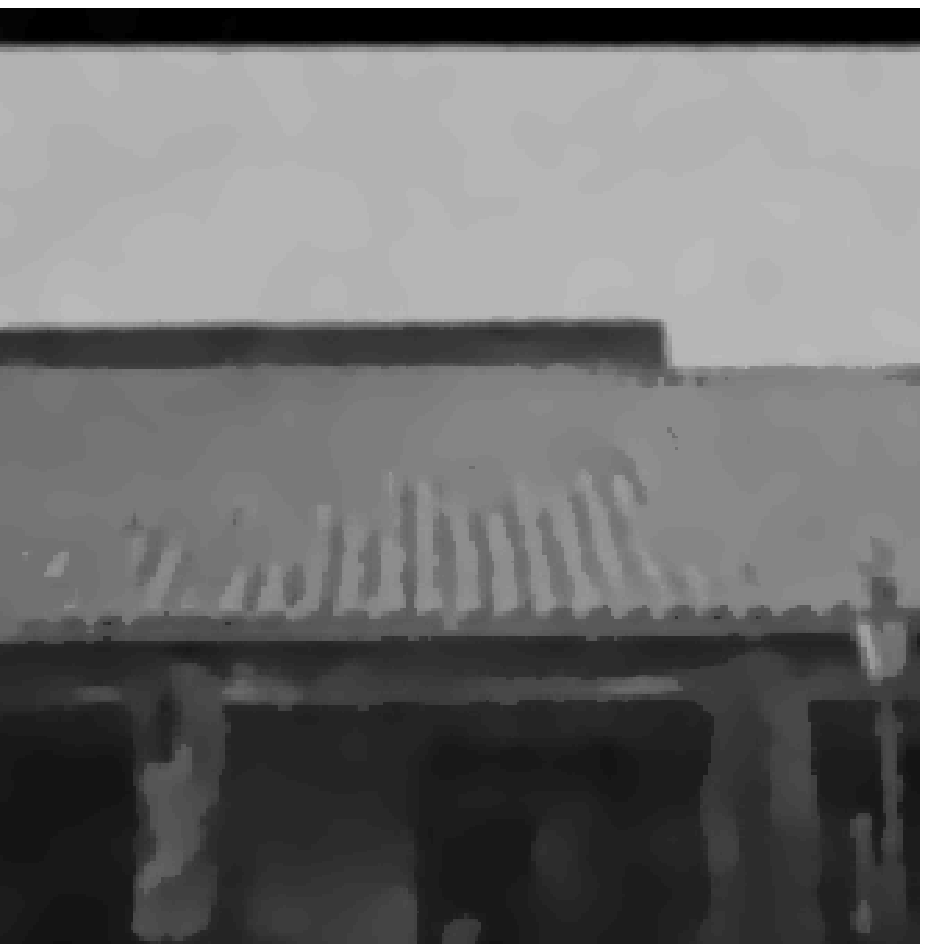}

\vspace{0.01\imwidth}

\includegraphics[width =
0.4\imwidth]{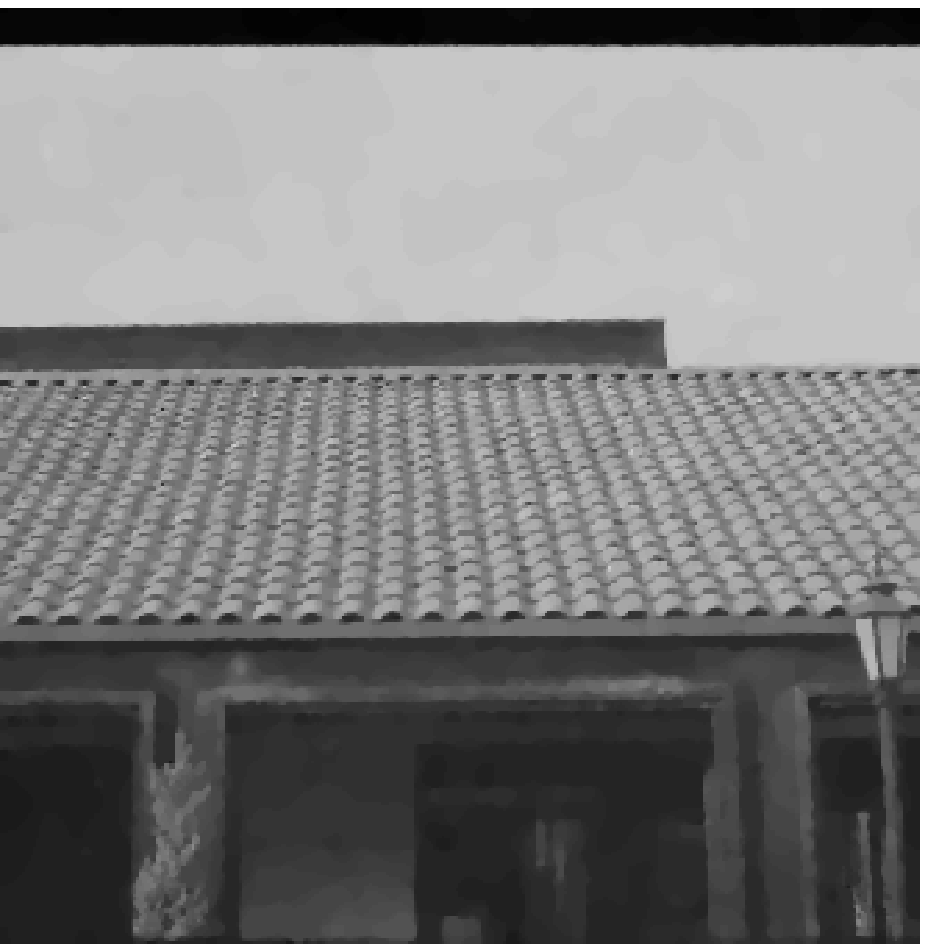}
\caption{Reconstructions ``roof'' (row-wise from top left):
$\Ls{2}$-oracle $\hat u_{\Ls{2}}$, Bregman-oracle $\hat u_B$, AWS estimators
$\hat u_{\text{aws}}^{\text{Triangle}}$ and $\hat u_{\text{aws}}^{\text{Gaussian}}$,
and SMRE $\hat u_{0.9}$.}\label{appl:resultsroof}
\end{center}
\end{figure}   

We evaluate the performance of the SMREs by means of a simulation
study. To this end, we compute the MISE, MIAE and MSB and compare these values
with the reference estimators. We note, however, that in particular the MISE
and MIAE are not well suited in order to measure the distance of images for
they are inconsistent with human eye perception. In \cite{frick:WanBovSheSim04}
the \emph{structural similarity index (SSIM)} was introduced for image quality    
assessment that takes into account luminance, contrast and structure of the
images at the same time. We use the author's implementation \footnote{available 
at \url{https://www.ece.uwaterloo.ca/~z70wang/research/ssim/}} which is
normalized such that the SSIM lies in the interval $[-1, 1]$ and is $1$ in case
of a perfect match. We denote by MSSIM the empirical mean of the SSIM in our
simulations. 

In Table \ref{appl:sim2} the simulation results are listed. A first striking
fact is the good performance of the $\Ls{2}$-oracle w.r.t.\ the MISE and
MIAE which is supposed to imply reconstruction properties superior to the
other methods. Keeping in mind the visual comparison in Figures
\ref{appl:resultscamera} and \ref{appl:resultsroof}, however, this is rather
questionable. On the other hand, it becomes evident that the $\Ls{2}$-oracle
has a rather poor performance w.r.t. the MSB which is more suited for measuring
image distances. It is therefore remarkable that the SMRE performs equally good
as the Bregman-oracle which, in contrast to the SMRE, is not accessible (since
$u^0$ is usually unknown). As far as the structural similarity measure MSSIM is
concerned our approach proves to be superior to all others. Finally, the
simulation results indicate that aws estimation is not  favourable for
denoising of natural images.

\begin{table}
{\scriptsize
\begin{center} 
\begin{tabular}{|l|c|c|c|c||c|c|c|c|}
\hline
 & \multicolumn{4}{|c||}{``cameraman''} & \multicolumn{4}{|c|}{``roof''} \\
 \hline\hline
 &  MISE  & MSB & MIAE  & MSSIM & MISE & MSB & MIAE & MSSIM \\
 \hline
 $\hat u_{\Ls{2}}$ & 0.0017 & 0.0314 & 0.0276  & 0.7739 & 0.0029 &
 0.0499  & 0.0383  & 0.6700   \\ 
 $\hat u_{\text{B}}$ & 0.0023 & 0.0256 & 0.0275  & 0.7995 &
 0.0038 & 0.0405  & 0.0391   & 0.6607 \\ 
 $\hat u_{\text{aws}}^{\text{Triangle}}$ & 0.0032 & 0.0482 & 0.0308   & 0.7657
 & 0.0046 & 0.0702 & 0.0416  & 0.6205   \\ 
 $\hat u_{\text{aws}}^{\text{Gauss}}$ & 0.0046 & 0.0470 & 0.0360 & 
 0.7284 & 0.0053 & 0.0686 & 0.0457  & 0.5668 \\
 \hline
 $\hat u_{0.9}$ & 0.0021 & 0.0252 & 0.0297  & 0.8024 & 0.0033  & 0.0374 
 & 0.0407   & 0.7003 \\
 \hline
\end{tabular}
\caption{Simulation studies for the test images ``cameraman'' and
``roof''.}\label{appl:sim2}
\end{center}}
\end{table}

\subsubsection*{Notes on the choice of $\Lambda$ and $\omega^S$.} 

In general, a proper choice of the transformation $\Lambda$ and of the
weight-functions $\omega^S$ can be achieved by including prior structural
information on the true image to be estimated. Substantial parts of natural images, such
as photographs, consists of oscillating patterns (as e.g. fabric, wood, hair,
grass etc.). This becomes obvious in the standard test images depicted in
Figure \ref{appl:test_images}. We claim that for signals that exhibit
oscillating patterns, a quadratic transformation $\Lambda$ as in
\eqref{appl:lambdasq} is favorable, since it yields (compared to the linear
statistic studied in Section \ref{appl:reg}) a larger power of the local test
statistic on small scales. 

In order to illustrate this, we simulate noisy observations $Y$ of the test
images $u$ in Figure \ref{appl:test_images} as in \eqref{appl:regr} with $\sigma
= 0.1$  and compute a \emph{global} estimator $\hat u$ by computing a minimizer
of the ROF-functional \eqref{appl:rof} (with $\lambda = 0.1$). We intend to
examine how well over-smoothed regions in $\hat u$ are detected by the
MR-statistic $T(Y - \hat u)$ as in \eqref{intro:mrstateqn} with two different
average functions (cf. \eqref{intro:mean})
\begin{equation*}
  \mu_{1,S}(v) = \abs{\sum_{\vec\nu \in S} v_{\nu}} \quad\text{ and }\quad
  \mu_{2,S}(v) = \sum_{\vec\nu \in S} v_{\nu}^2
\end{equation*}
respectively. For the sake of simplicity we restrict for the moment our
considerations on the index set $\S$ of all $5\times 5$ sub-squares in
$\set{1,\ldots,m}^2$. In Figure \ref{appl:local_means} the local means
$\mu_{i,S}$  of the residuals $v = Y - \hat u$ for the
``roof''-image are depicted. To be more precise, the center coordinate of each
square $S\in\S$ is colored according to $\mu_{i,S}$, Hence, large values
indicate locations where the estimator $\hat u$ is considered over-smoothed
according to the statistic. It becomes visually clear that the localization of
oversmoothed regions is better for $\mu_{2,S}$. This is a good motivation for
incorporating the local means of the squared residuals in the SMRE model \eqref{intro:smreeqnp}. 

\begin{figure}[!ht]
\begin{center}
\includegraphics[height = 0.35\imwidth]{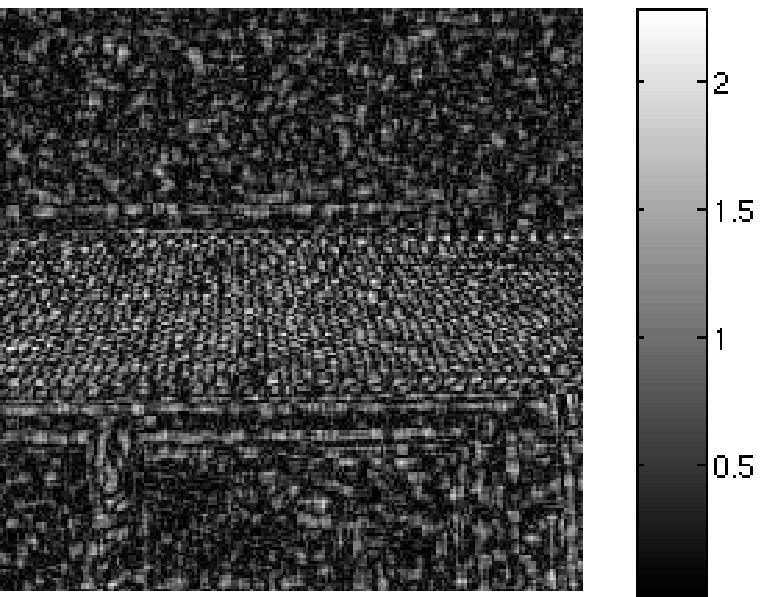}
\hspace{0.01\imwidth}
\includegraphics[height = 0.35\imwidth]{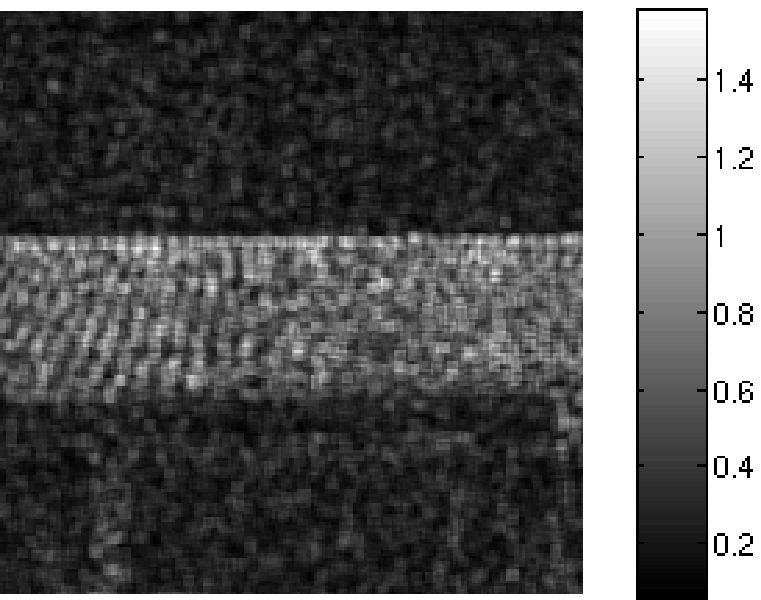}
\caption{Local means $\mu_{1,S}$ (left) and $\mu_{2,S}$ (right)) of the
residuals for ``roof'' image.}\label{appl:local_means}
\end{center}
\end{figure}

We finally comment on the choice of $c_S$. Since $\eps_{\vec\nu}$ are
independent and normally distributed random variables, the (scaled) average function 
\begin{equation*}
  \sigma^{-2}\mu_S(\eps) = \sum_{\vec\nu \in S}
  \left(\frac{\eps_{\vec\nu}}{\sigma}\right)^2
\end{equation*}
is $\chi^2$ distributed with $\# S$ degrees of freedom. Note that the
distribution of $\sigma^{-2}\mu_S(\eps)$ is identical only for sets $S$ of the
same scale $\# S$. As a consequence of this, it is likely that
certain scales dominate the supremum in the MR-statistic $T$ which spoils the
multiscale properties of our approach. As a way out, we compute normalizing
constants \emph{for each scale separately.} 

An alternative approach would be to search for transformations that turn
$\mu_S(\eps)$ into almost identically distributed random variables. Logarithmic
and $p$-root transformations are often employed for this purpose (see e.g.
\cite{frick:HawWix86}). This will be investigated separately.

\subsubsection*{Implementation Details.} 

The current index set $\S$ results in an overall number of 
constraints in \eqref{appl:regrsmre2d} of
\begin{equation*}
  \# \S = \sum_{i=1}^{25} (512-i+1)^2 = 6251300.
\end{equation*} 
Again by grouping independent side-conditions, the
system $\mathcal{S}$ can be grouped such that the intersection of the 
corresponding sets $D_1,\ldots,D_M$ in \eqref{impl:regroup} form
$\mathcal{C}$  with
\begin{equation*}
  M = \sum_{i=1}^{25} i^2 = 5525.
\end{equation*}
In all our simulations we set $\tau = 10^{-4}$ and $\lambda = 0.25$ in Algorithm
\ref{impl:ala} which results in $k[\tau] \approx 30$ iterations and a overall
computation time of approximately $2$ hours for each SMRE. Hence,
parallelization is clearly desirable in this case. 

\subsection{Deconvolution}\label{appl:deblurring}

Another interesting class of problems which can be approached by means of
SMREs are deconvolution problems. To be more precise, we assume that
$K$ is a convolution operator, that is
\begin{equation*}
  (Ku)_{\vec\nu} = (k \ast u)_{\vec\nu} = \sum_{\vec m \in \R^d} k_{\vec\nu - \vec
  m} u_{\vec m}
\end{equation*}
where $k$ is a square-summable kernel on the lattice $\Z^d$ and $u\in H$ is
extended by zero-padding. We will focus on the situation where $k$ is a
circular Gaussian kernel with standard deviation $\sigma$ given by
\begin{equation}\label{appl:defgauss}
  k_{\vec\nu} = \frac{1}{(\sqrt{2 \pi} \sigma)^d} e^{-
  \frac{\sum_{i=1}^d \nu_i^2}{2 \sigma^2}}.
\end{equation}
With $Z = Y + \lambda p_{k-1}+v_k$, the
primal step \eqref{ala:primal} in Algorithm \ref{impl:ala} amounts to solve
\begin{equation*}
  u_k \leftarrow \argmin_{u\in H} \frac{1}{2}\sum_{\vec\nu\in
  X}((Ku)_{\vec\nu} - Z_{\vec\nu})^2 + \lambda J(u),
\end{equation*}
where we choose $J$ to be as in \eqref{appl:tvpenalty} and apply the
techniques described in \cite{frick:Vog02} for the numerical solution.

In order to illustrate the performance of our approach in practical applications,
we give an example from confocal microscopy, nowadays a standard technique in
fluorescence microscopy (cf. \cite{frick:Paw06}). When recording images with
this kind of microscope, the original object gets blurred by a Gaussian kernel (in first order). The observations
(photon counts) can be modeled as a Poisson process, i.e.
\begin{equation}\label{appl:poiss}
  Y_{\vec\nu} = \text{Poiss}((Ku^0)_{\vec\nu}),\quad \vec\nu \in X.
\end{equation}  

The image depicted in Figure \ref{appl:cytodata} shows a recording of a PtK2 cell
taken from the kidney of \emph{potorous tridactylus}. Before the recording, the protein $\beta$-tubulin
was tagged with a fluorescent marker such that it can be traced by the
microscope. The image in \ref{appl:cytodata} shows an area of $18\times 18$
$\tcmu \text{m}^2$ at a resolution of $798\times 798$ pixel. The point spread
function of the optical system can be modeled as a Gaussian kernel with full
width at half maximum of $230$nm,  which corresponds to $\sigma = 4.3422$ in
\eqref{appl:defgauss}.  

Note that \eqref{appl:poiss} does not fall immediately into the range of models  
covered by \eqref{intro:lineqn}. We will adapt the present situation to the SMRE
methodology described in Section \ref{intro} by standardization and consider
instead of \eqref{intro:smreeqn} the modified problem
\begin{equation}\label{appl:smreeqnpoiss}
  \inf_{u\in U} J(u) \quad\text{ s.t.}\quad
  T\left(\frac{Y-Ku}{\sqrt{Ku}}\right)\leq 1
\end{equation}
where the division is understood pointwise. Clearly, the problem of finding a
solution of \eqref{appl:smreeqnpoiss} is much more involved than solving
\eqref{intro:smreeqn} for the constraints being \emph{nonconvex}: firstly, the
functional $G$ as defined in \eqref{impl:feasible} is nonconvex as a
consequence of which the convergence result in Theorem \ref{impl:alaconv} does
not apply and secondly Dykstra's projection algorithm as described in Section
\ref{impl:dyk} cannot be employed.  

We propose the following ansatz in order to circumvent this problem: instead of
projecting onto the intersection $\mathcal{C}$ of sets $C_S$ as described in
\eqref{impl:sidecond}, we now project in the $k$-th step of Algorithm \ref{impl:ala}
onto
\begin{equation*}%\label{appl:poisproj}
  \mathcal{C}_P[k] = \bigcap_{n=1}^N C_{P,S}[k] \quad \text{where} \quad
  C_{P,S}[k] = \set{ v \in H~:~ \mu_{S}\left(v\slash \sqrt{Ku_k}\right)
  \leq q}.
\end{equation*}
with a pointwise division by the square root of $K u_k$. Put differently, in
the $k$-th step of Algorithm \ref{impl:ala} we use the previous estimate $u_k$ of
$u^0$ as a \emph{lagged standardization} in order to approximate the
constraints in \eqref{appl:smreeqnpoiss}. In fact, we use $\sqrt{\max(Ku_k,
\eps)}$ with a small number $\eps>0$ for standardization, in order to avoid
instabilities.

We note that while with this modification Dykstra's algorithm becomes applicable
again, the projection problem \eqref{ala:noise} now changes in each iteration
step of Algorithm \ref{impl:ala}. As a consequence, Theorem \ref{impl:alaconv}
does not hold anymore after this modification, either. So far, we have not come
up with a similar convergence analysis.

We compute the SMRE $\hat u_{0.9}$ by employing Algorithm \ref{impl:ala} with
the modifications described above. As in the denoising examples in Section
\ref{appl:denoising} the index set $\S$ consists of all squares with the
side-lengths $\set{1,\ldots,25}$ and we choose $\omega^S = \chi_S$ and $\Lambda =
\id$. We note, that this results in an overall number of $\#\S = 95~436~200$
inequality constraints. The constant $q$ are chosen as in
\eqref{appl:quantile}, where we assume that $\eps_{\vec\nu}$ are independent and standard normally distributed r.v. 

In Algorithm \ref{impl:ala} we set $\lambda = 0.05$ and compute $100$ steps. We
observe that after a few iterations ($\sim 15$) the error $\tau$ falls below
$10^{-3}$ and almost stagnates thereafter, which is due to the fact that we do
not increase the accuracy in the subroutines for \eqref{ala:primal} and
\eqref{ala:noise}. Each iteration step in Algorithm \ref{impl:ala} approximately
takes $10$ minutes, where $90\%$ of the computation time is needed for
\eqref{ala:primal}. The result is depicted in Figure \ref{appl:cytoresult}.

\begin{figure}[h!]
\begin{center}
\subfigure[Fluorescence microscopy data of a PtK2 cell in \emph{potorous
tridactylus} kidney. The bright filaments indicate the location of the protein
$\beta$-tubulin.]{\includegraphics[height =
0.4\imwidth]{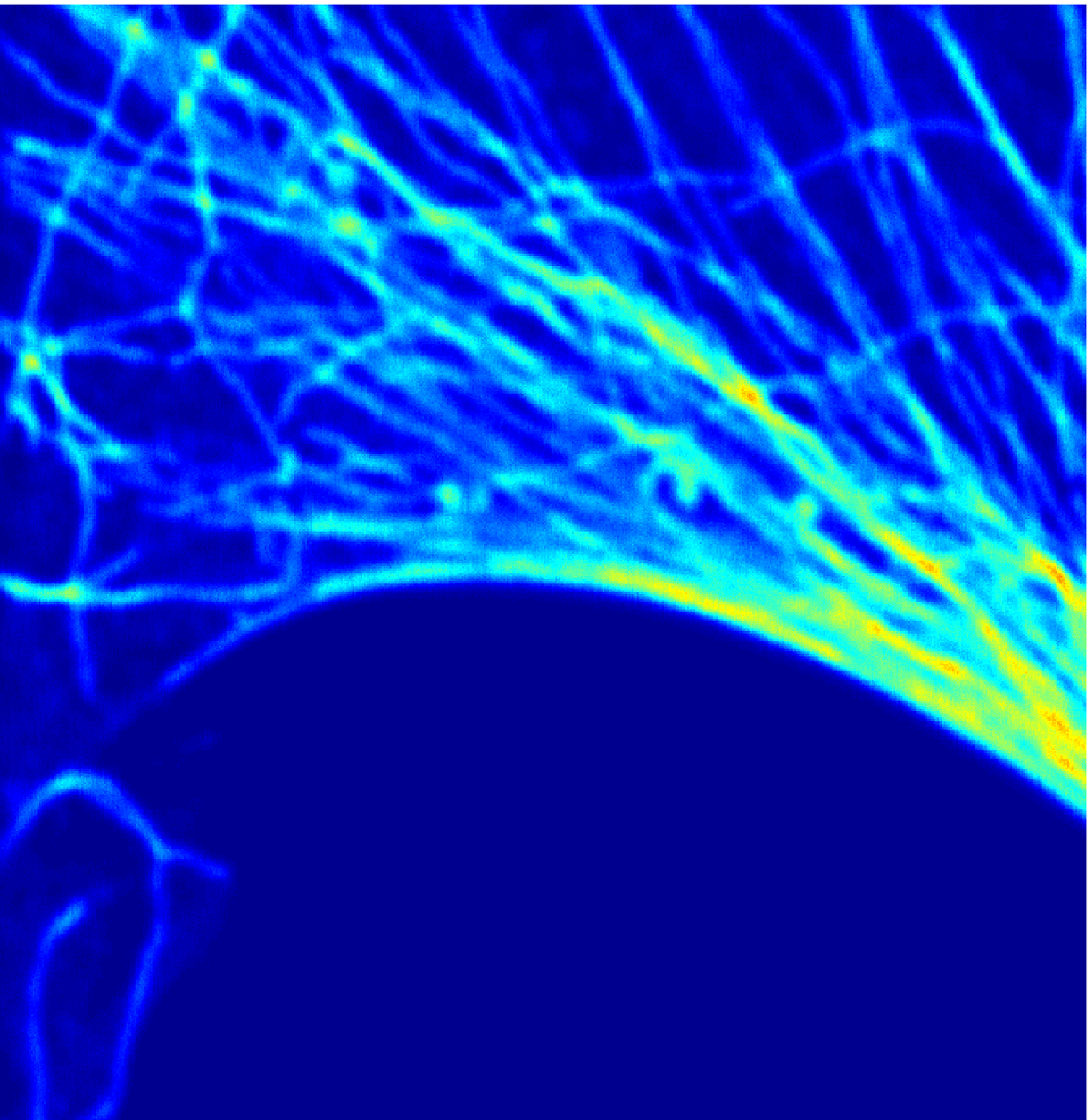}\label{appl:cytodata}}
\hspace{0.05\imwidth}
\subfigure[SMRE $\hat u_{0.9} $: fully automated and locally adaptive
deconvolution of microscopy data.]{\includegraphics[height =
0.4\imwidth]{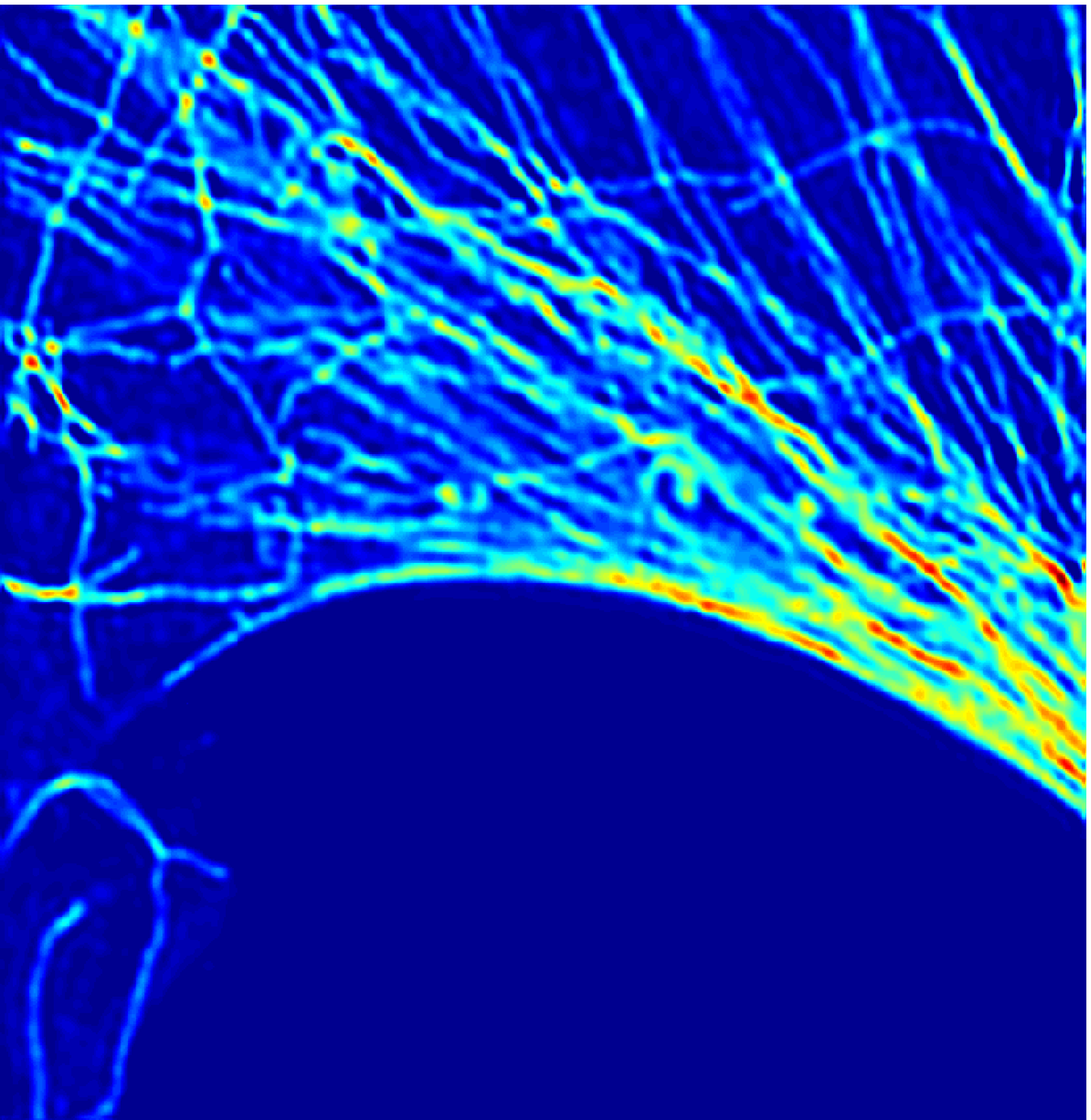}\label{appl:cytoresult}}
\caption{Reconstruction of confocal microscopy data.}
\end{center}
\end{figure}

The benefits of our method are twofold:
\begin{enumerate}[i)]
  \item The amount of regularization is chosen in a \emph{completely automatic
  way}. The only parameter to be selected is the level $\alpha$ in \eqref{appl:quantile}. Note that the parameter $\lambda$ in Algorithm
  \ref{impl:ala} has no effect on the output (though it has an effect on the number of iterations needed and the numerical stability).
  \item The reconstruction has an appealing locally adaptive behavior which in
  the present example mainly concerns the gaps between the protein filaments:
  whereas the marked $\beta$-tubulin is concentrated in regions of basically
  one scale, the gaps in between actually make up the multiscale nature of the
  image.
\end{enumerate}

\begin{figure}[h!]
\begin{center}
\subfigure[STED microscopy recording of the PtK2 cell
data set.]{\includegraphics[height =
0.4\imwidth]{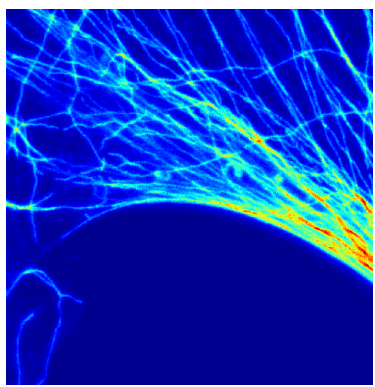}\label{appl:cytosted}}
\hspace{0.05\imwidth}
\subfigure[Detail comparison between confocal recording
(left), SMRE $\hat u_{0.9} $ (middle) and STED recording (right).
]{\includegraphics[height =
0.4\imwidth]{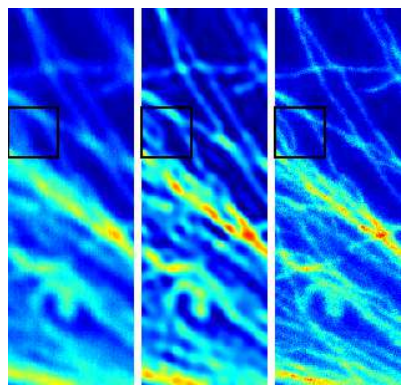}\label{appl:cytocomp}}
\caption{Comparison with high-resolution STED microscopy data.}
\end{center}
\end{figure}   

In the present situation we are in the comfortable position to have a reference
image at hand by means of which we can evaluate the result of our method: STED
(STimulated Emission Depletion) microscopy constitutes a relatively new method, that is
capable of recording images at a physically $5$-$10$ times higher resolution as
confocal microscopy (see \cite{frick:Hel94,frick:Hel07}). Hence a STED
image of this object may serve as ``gold standard'' reference image. 

Figure \ref{appl:cytosted} depicts a STED recording of the PtK2 cell data set in
Figure \ref{appl:cytodata}. The comparison of the SMRE $\hat u_{0.9}$ with the
STED recording in Figure \ref{appl:cytocomp} shows that our SMRE technique
chooses a reasonable amount of regularization: no artifacts due to
under-regularization are generated and on the other hand almost all relevant
geometrical features that are present in the high-resolution STED recording
become visible in the reconstruction. In particular, we note that filament
bifurcations (one such bifurcation is marked by a black box in Figure
\ref{appl:cytocomp}) become apparent in our reconstruction that are not visible
in the recorded data.

Finally, we mention that aside to standardization, other transformations of the
Poisson data \eqref{appl:poiss} could possibly be considered. For example
\emph{Anscombe's transformation} is known to yield reasonable approximations to
normality even for low Poisson-intensities and hence has a particular
appeal for e.g. microscopy data with low photon-counts. We are currently
investigating SMREs that employ Anscombe's transform, where in particular the
arising projection problems are challenging.

\section{Conclusion and Outlook}\label{outlook} 

In this work, we propose a general estimation technique for nonparametric
inverse regression problems in the white noise model \eqref{intro:lineqn} based
on the convex program \eqref{intro:smreeqnp}. It amounts to finding a
minimizer of a convex regularization functional $J(u)$ over a set of feasible
estimators that satisfy the fidelty condition $T(Y-Ku)\leq q$, where $T$ is
assumed to be the maximum over simple convex constraints and $q$ is some
quantile of the statistic $T(\eps)$. Any such minimizer we call
\emph{statistical multiresolution estimator (SMRE)}. This approach covers well
known uni-scale techniques, such as the Dantzig selector, but with a vast field
of potentially new application areas, such as locally adaptive imaging. The
particular appeal of the multi-scale generalization arises for those situations
where a ``neighboring relationship'' within the signal can be employed to
gain additional information by ``averaging'' neighboring residuals. We
demonstrate in various examples that this improvement is drastic. 

We approach the numerical solution of the problem by the ADMM (cf. Algorithm
\ref{impl:ala}) that decomposes the problem into two subproblems: A
$J$-penalized least squares problem, independent of $T$, and an orthogonal
projection problem onto the feasible set of \eqref{intro:smreeqnp} that is
independent of $J$. The first problem is well studied and for most typical
choices of $J$ fast and reliable numerical approaches are at hand. The
projection problem, however, is computational demanding, in particular for image
denoising applications. We propose Dykstra's cyclic projection method for its
approximate solution. Finally, by extensive numerical experiments, we illustrate
the performance of our estimation scheme (in nonparametric regression, image
denoising and deblurring problems) and the applicability of our algorithmic
approach.

Summarizing, this paper is meant to introduce a novel class of statistical
estimators, to provide a general algorithmic approach for their numerical computation and to
evaluate their performance by numerical simulations. The inherent questions on
the asymptotic behaviour of these estimators (such as consistency, convergence rates or oracle
inequalities) remain ---to a large extent--- unanswered. This opens an
interesting area for future research.

A first attempt has been made in \cite{frick:FriMarMun10} where it is assumed
that the model space $U\ni u^0$ is some Hilbert-space of real valued functions on some domain $\Omega$ and that
$K:U\ra\L{2}$ is linear and bounded. The error model \eqref{intro:lineqn} then
has to be adapted accordingly. When $Y$ is a Gaussian process on $\L{2}$ with
mean $Ku^0$ and variance $\sigma^2>0$, consistency and
convergence rates for SMREs as $\sigma\ra 0^+$ have been proved in
\cite{frick:FriMarMun10} for the case when $\Lambda = \id$. However, in order to
extend these results to the present setting, one would rather work with a
discrete sample of $Ku_0$ on the grid $X$ and then consider the case when the
number of observations $N = md$ tends to infinity. The previous analysis in
\cite{frick:FriMarMun10} indicates two major aspects that have to be considered
in the asymptotic analysis for SMREs: \begin{enumerate}[(a)]
  \item As $N\ra\infty$ usually the cardinality of the index set $\mathcal{S}$
  (and hence of the set of weight functions $\mathcal{W}$) gets unbounded. Thus,
  the mutliresolution statistic $T(\eps) = T_N(\eps)$ in \eqref{intro:mrstateqn}
  is likely to degenerate unless it is properly normalized and $\W$ satisfies
  some \emph{entropy condition}. In the linear case ($\Lambda = \id$) we
  utilized a result from \cite{frick:DueSpo01} that guarantees a.s. boundedness
  of $T_N(\eps)$. \item In order to derive convergence rates (or risk bounds) it
  is well known that the true signal $u^0$ has to satisfy some apriori
  regularity conditions. When using general convex regularization functionals
  $J$, this is usually expressed by the \emph{source condition}
  \begin{equation*}
  K^* p^0 \in \partial J(u^0), \text{ for some } p^0 \in \L{2}.
  \end{equation*}
  Here $K^*$ denotes the adjoint of $K$ and $\partial J$ the (generalized)
  derivative of $J$. For example, if $J(u) = \frac{1}{2}\norm{u}^2$, then this
  conditions means that $u^0 \in \ran(K^*)$. 
\end{enumerate}
It would be of great interest to transfer and extend the results in
\cite{frick:FriMarMun10} to the present situation. It is to be expected that
 (a) and (a) above are necessary assumptions for this purpose. 
 
 As stressed by the referees, other extensions are of interest and will be
 postponed to future work. In contrast to imaging, in many other applications
 the design $X$ is random, rather than fixed. In these situations an obvious way
 to extend our algorithmic framework would be to select suitable partitions $\S$
 according to the design density, i.e. with finer resolution at locations with a
 high concentration of design points. It also remains an open issue how to
 extend the SMRE methodology to density estimation rather than regression, in
 particular in a deconvolution setup. For $d=1$ a first step in this direction
 has been taken in \cite{frick:DavKov04} and it will be of great interest to
 explore whether our approach allow this to be extended to $d\geq 2$.
 
\section*{Acknowledgement}

K.F. and A.M. are supported by the DFG-SNF Research Group FOR916
\emph{Statistical Regularization and Qualitative constraints}. P.M is supported
by the BMBF project $03$MUPAH$6$ \emph{INVERS}. A.M  and P.M. are supported by the
SFB755 \emph{Photonic Imaging on the Nanoscale} and the SFB803
\emph{Functionality Controlled by Organization in and between Membranes}. 

We thank S.~Hell, A.~Egner and A.~Schoenle  (Department of
NanoBiophotonics, Max Planck Institute for Biophysical Chemistry, 
G{\"o}ttingen) for providing the microscopy data and L.~D{\"u}mbgen
(University of Bern) for stimulating discussions. Finally, we are grateful to an
associate editor and to an anonymous referee for their helpful comments. 
 
\appendix

\section{Proofs}\label{app}

In this section we shall give the proofs of Theorems \ref{impl:alaconv} and
\ref{impl:alaconvcor} as well as Corollary \ref{impl:alaconvcortwo}. We note,
that convergence of Algorithm \ref{impl:ala} is a classical subject in
optimization theory and a proof can e.g. be found in \cite[Chap. III Thm.
4.1]{frick:FG83}. However, in order to apply these results, it is necessary that
certain regularity conditions for $J$ hold, that are not realistic for our
purposes (as e.g. in the case of total-variation regularization). The assertions
of Theorems \ref{impl:alaconv} and \ref{impl:alaconvcor} are modifications of
the standard results. 

Moreover, we will allow for \emph{approximate solution} of the subproblems
\eqref{ala:noise} and \eqref{ala:primal}. To this end, we rewrite these two
subproblems as variational inequalities, i.e. given $(u_{k-1}, v_{k-1},
p_{k-1})$ we find $(u_k, v_k, p_k)$ such that
\begin{subequations}
\begin{gather}
G(v) - G(v_k) + \lambda^{-1}\inner{Ku_{k-1} + v_k - Y - \lambda p_{k-1}}{v -
v_k} \geq -\eps_k,\;\forall v\in H \label{app:noise}\\
J(u) - J(u_k) + \lambda^{-1}\inner{Ku_k
+ v_k - Y - \lambda p_{k-1}}{Ku - Ku_k} \geq
-\delta_k,\;\forall u\in U \label{app:primal}\\ 
p_k = p_{k-1} - (Ku_k + v_k - Y)\slash \lambda\label{app:dual},
\end{gather}
\end{subequations}
where we assume that $\set{\eps_1,\eps_2,\ldots}$ and
$\set{\delta_1,\delta_2,\ldots}$ are given sequences of positive numbers. Note
that \eqref{app:noise} implies that $G(v_k) = 0$ and hence
$v_k\in\mathcal{C}$ and that for $\eps_k = \delta_k = 0$ \eqref{app:noise} and
\eqref{app:primal} are equivalent to \eqref{ala:noise} and \eqref{ala:primal}, respectively.

Finally, we remind the reader of the definition of the \emph{subdifferential}
(or generalized derivative) $\partial F$ of a convex function $F:V\ra \R$ on a real Hilbert-space $V$:
\begin{equation*}
  \xi\in\partial F(v)\quad\Leftrightarrow\quad F(w)\geq F(v) +
  \inner{\xi}{w-v}_V\;\forall(w\in V). 
\end{equation*}
If $\xi\in\partial F(v)$, then $\xi$ is called \emph{subgradient} of $F$ at
$v$. It follows from \cite[Chap III, Prop. 3.1 and Prop. 4.1]{frick:ET76}
that the Lagrangian $L$ (and hence also the augmented Lagrangian $L_\lambda$)
has a saddle-point $(\hat u, \hat v, \hat p)\in U\times H\times H$ if and only if
\begin{equation}\label{impl:kkt}
  K\hat u + \hat v = Y,\quad K^*\hat p\in \partial J(\hat u)\quad\text{ and
  }\quad \hat p\in \partial G(\hat v).
\end{equation} 
We will henceforth assume that $\set{(u_k, v_k, p_k)}_{k\in\N}$ is a
  sequence generated by iteratively repeating the steps \eqref{app:noise} -
  \eqref{app:dual}. Further, we introduce the notation
  \begin{equation*}
    \bar u_k := u_k - \hat u,\quad\bar v_k := v_k - \hat v\quad\text{ and }\quad
    \bar p_k := p_k - \hat p.
  \end{equation*}
  We start with the following 
  
\begin{lem}\label{app:lemma}
For all $k\geq 1$ we have that
\begin{multline}\label{app:aux85}
  \left( \norm{\bar p_{k-1}}^2 + \lambda^{-2} \norm{K\bar u_{k-1}}^2 \right) -
  \left( \norm{\bar p_k}^2 + \lambda^{-2} \norm{K\bar u_k}^2 \right)
  \\ \geq \lambda^{-2}\left( \norm{K\bar u_k + \bar v_k}^2 + \norm{K \bar
  u_{k-1} - K \bar u_k}^2\right) - 2\lambda^{-2}(\delta_k + \delta_{k-1}) -
  \delta_k - \eps_k
\end{multline}
\end{lem}
\begin{proof}
The assertion follows by repeating the steps (5.6)-(5.25) in the proof of
\cite[Chap. III Thm. 4.1]{frick:FG83} after replacing (5.9) and (5.10)
by \eqref{app:noise} and \eqref{app:primal} respectively.
\end{proof}

We continue with the proof of Theorem \ref{impl:alaconv}. More precisely, we
prove the following generalized version
\begin{prop}\label{app:alaconv}
Assume that the sums $\sum_{k=1}^\infty \delta_k$ and $\sum_{k=1}^\infty \eps_k$
are finite. Then, the sequence $\set{(u_k, v_k)}_{k\geq1}$ is bounded in
$U\times H$  and every weak cluster point is a solution of \eqref{impl:linconstr}. Moreover,
  \begin{equation*}
    \sum_{k\in\N} \norm{Ku_k + v_k - Y}^2 + \norm{K(u_k - u_{k-1})}^2 < \infty.
  \end{equation*}
\end{prop}
        
\begin{proof} Let $k\geq1$ and define $D =
\sum_{k=1}^\infty \delta_k$ and $E = \sum_{k=1}^\infty \eps_k$. 
Summing up Inequality \eqref{app:aux85} over $k$ and
keeping in mind that $K\bar u_k + \bar v_k = K u_k + v_k -Y$ and $K \bar u_{k-1} - K \bar u_k = K u_{k-1} -  K u_k$ shows
\begin{multline*}%\label{app:aux9}
  \sum_{k=1}^\infty  \norm{K u_k + v_k - Y}^2 + \norm{K u_{k-1} - K u_k}^2 
  \\ \leq \lambda^2\norm{\hat p}^2 + \norm{K \hat u}^2 + (4\lambda^{-2}+1)D + E<
  \infty
\end{multline*}
where we have used that $\bar u_0 = \hat u$ and $\bar p_0 = \hat p$. 
Furthermore, it follows again from \eqref{app:aux85}  that 
\begin{equation}\label{app:aux9}
\norm{\bar p_k}^2 + \lambda^{-2}\norm{K\bar u}^2 \leq \norm{\hat p}^2 +
\lambda^{-2}\norm{K\hat u}^2 + (4\lambda^{-2}+1)D + E < \infty
\end{equation}
This together with  the fact that $\norm{K u_k + v_k - Y} \ra 0$
shows that
\begin{equation*}
  \max( \norm{Ku_k}, \norm{v_k}, \norm{p_k} ) = \bigo(1).
\end{equation*}
Together with the optimality condition for \eqref{ala:primal} this in
turn implies that for an arbitrary $u\in H$
\begin{equation}\label{app:opt}
  J(u_k) \leq J(u) + \lambda^{-1} \inner{K u_k + v_k - Y - \lambda
  p_{k-1}}{K u - K u_k} + \delta_k = \bigo(1).
\end{equation}
Summarizing, we find that 
\begin{equation*}
  \max_{S\in \S} \mu_{S}(Ku_k-Y)+J(u_k) \leq \max_{S\in
  \S} \norm{\omega^S}\norm{\Lambda(Ku_k-Y)} + J(u_k) \leq c< \infty
\end{equation*}
for a suitably chosen constant $c\in\R$, since $\Lambda$ is supposed to be
continuous. Thus, it follows from Assumption \ref{review:assex} that
$\set{u_k}_{k\in\N}$ is bounded and hence sequentially weakly compact. Now, let
$(\tilde u, \tilde v, \tilde p)$ be a weak cluster point of $\set{(u_k,v_k,p_k)}_{k\in\N}$ and recall that $(\hat u, \hat v, \hat p)$ was
assumed to be a saddle point of the augmented Lagrangian $L_\lambda$. Setting $u
= \hat u$ in \eqref{app:opt} thus results in
\begin{equation}\label{app:aux10}
  \begin{split}
  J(u_k) & \leq J(\hat u) + \lambda^{-1}\inner{K u_k + v_k - Y}{K\hat u
  - K u_k} + \inner{p_{k-1}}{K u_k - K \hat u} + \delta_k\\
  &  = J(\hat u) +
  \inner{p_{k-1}}{K u_k - K \hat u} + \smallo(1)
\end{split}
\end{equation}
Using the relation $K\hat u  + \hat v = Y$ we
further find
\begin{multline}\label{app:aux11}
  \inner{p_{k-1}}{K u_k - K \hat u} =  \inner{p_{k-1}}{K u_k - Y + \hat v} \\
  = \inner{p_{k-1}}{K u_k +v_k - Y} - \inner{p_{k-1}}{ v_k - \hat v} =
  \smallo(1) - \inner{p_{k-1}}{ v_k - \hat v}
\end{multline}
From the definition of $v_k$ in \eqref{app:noise} and from the fact that $\hat
v,v_k \in \mathcal{C}$ it follows that
\begin{equation*}
  \inner{Y + \lambda p_{k-1} - (K u_{k-1} + v_k)}{\hat v - v_k} \leq \eps_k
\end{equation*}
which in turn implies that
\begin{multline}\label{app:aux12}
  - \inner{p_{k-1}}{ v_k - \hat v} \leq \lambda^{-1} \inner{Y  - (K u_{k-1}
  + v_k)}{v_k- \hat v} + \eps_k \\ = \lambda^{-1} \inner{Y  - (K u_{k}
  + v_k)}{v_k- \hat v} + \lambda^{-1} \inner{Ku_k - K u_{k-1}}{v_k- \hat v} +
  \eps_k = \smallo(1)
\end{multline}
Combining \eqref{app:aux10}, \eqref{app:aux11} and \eqref{app:aux12} gives
\begin{equation*}
  \limsup_{k\ra\infty} J(u_k) \leq J(\hat u).
\end{equation*}
Now, choose a subsequence $\set{u_{\rho(k)}}_{k\in\N}$ such that $u_{\rho(k)}
\rightharpoonup \tilde u$. Since $J$ is convex and lower semi-continuous it is
also weakly lower semi-continuous and hence the previous estimate yields
\begin{equation*}
  J(\tilde u) \leq\liminf_{k\ra\infty} J(u_{\rho(k)}) \leq J(\hat u).
\end{equation*}
Moreover, we have that $v_{\rho(k)} \in \mathcal{C}$ for all $k\in\N$. Since
$\mathcal{C}$ is closed  we conclude that $\hat v \in \mathcal{C}$.
Since $K\tilde u + \tilde v = Y$ this shows that $(\tilde u, \tilde v)$ solves
\eqref{impl:linconstr} and thus $J(\tilde u) = J(\hat u)$.
\end{proof}

We proceed with the proof of Theorem \ref{impl:alaconvcor}. Again we present a
generalized version. To this end, let $D$ and $E$ be as in Proposition
\ref{app:alaconv}. 

\begin{prop}\label{app:alaconvcor} 
There exists a constant $C = C(\lambda, \hat u, \hat v, \hat p, E, D)$ such that 
\begin{equation*}
  0\leq J(u[\tau]) -  J(\hat u) - \inner{K^*\hat p}{u[\tau] - \hat u }_U\leq 
  C \tau + \delta_{k[\tau]} + \eps_{k[\tau]}\quad\forall(\tau >
  0).
\end{equation*}
\end{prop} 

\begin{proof}
  Define $B^2 = (4\lambda^{-1}+1)D + E$. Then it follows from \eqref{app:aux9}
  that
  \begin{gather}
    \lambda^{-1}\norm{Ku_k - K\hat u} \leq \norm{\hat p} +
    \lambda^{-1}\norm{K\hat u} + B \label{app:aux456} \\
     \norm{p_k} \leq 2\norm{\hat p} + \lambda^{-1} \norm{K\hat u} +
     B,\label{app:aux789}
  \end{gather}
  where $(\hat u, \hat v, \hat p)$ is an arbitrary saddle point of
  $L_\lambda(u,v,p)$. Assume that $\tau > 0$ and that $k = k[\tau]$ is such that
  \begin{equation*}
    \max(\norm{K u_k + v_k - Y}, \norm{K u_{k-1} - K u_k}) \leq \tau.
  \end{equation*}
  Then, it follows from \eqref{app:aux10}, \eqref{app:aux12},
  \eqref{app:aux456} and \eqref{app:aux789} and that
  \begin{equation}\label{app:aux123}
      \begin{split}
    J(u_k) & \leq J(\hat u) + \lambda^{-1}\inner{K u_k + v_k - Y}{K\hat u
  - K u_k} + \inner{p_{k-1}}{K u_k - K \hat u} + \delta_k\\
  & \leq J(\hat u) + \tau(\norm{\hat p} +
    \lambda^{-1}\norm{K\hat u} + B) + \norm{p_{k-1}}\tau + \inner{p_{k-1}}{\hat
    v - v_k} + \delta_k \\ & \leq J(\hat u) + \tau(3\norm{\hat p} +
    2\lambda^{-1}\norm{K\hat u} + 2B) + 2\lambda^{-1}\tau
      \norm{v_k - \hat v} + \delta_k + \eps_k
  	\end{split}
  \end{equation}
  After observing that $Ku_k - K\hat u + v_k -
  \hat v = Y$ it follows that $\norm{v_k - \hat v} \leq \tau + \norm{Ku_k -
  K\hat u}$ and combining \eqref{app:aux123} and \eqref{app:aux456} gives
  \begin{equation*}
  J(u_k) \leq J(\hat u) + \tau\left( 5\norm{\hat p} + 4\lambda^{-1}\norm{K\hat
  u} + 4B +  2\lambda^{-1}\tau \right) + \delta_k + \eps_k.
  \end{equation*}
  Now, observe that from the definition of the subgradient and \eqref{impl:kkt},
  it follows that $J(u_k) \geq J(\hat u) + \inner{K^*\hat p}{u_k-\hat u}$
  and that $\inner{\hat p}{v_k -\hat v} \leq 0$. This and the fact that $K\hat u
  + \hat v = Y$ implies that 
  \begin{equation}\label{app:aux1000}
      \begin{split}
      0& \leq J(u_k)  -J(\hat u) - \inner{K^*\hat p}{u_k-\hat u} \\
      & = J(u_k)  -J(\hat u) - \inner{\hat p}{Ku_k + v_k -Y} + \inner{\hat
      p}{K\hat u + \hat v - Y} + \inner{\hat p}{v_k - \hat v} \\
      & \leq J(u_k)  -J(\hat u) + \norm{\hat p} \tau \\
      & \leq \tau\left( 6\norm{\hat p} + 4\lambda^{-1}\norm{K\hat
  u} + 4B +  2\lambda^{-1}\tau \right) + \delta_k + \eps_k.
      \end{split}
  \end{equation}
  This together with \eqref{app:aux123} finally proves the first part 
 of the assertion.
 \end{proof}
 
 \begin{rem}
 For the case when $D = E = 0$, it is seen from the proof of Proposition
 \eqref{app:alaconvcor} that the constant $C$ takes the simple form
 \begin{equation*}
 C = \tau\left( 6\norm{\hat p} + \frac{4\norm{K\hat
  u} + 2\tau}{\lambda} \right).
 \end{equation*} 
 \end{rem}
 
 \begin{proof}[Proof of Corollary \ref{impl:alaconvcortwo}]
 Assume that $J(u) = \frac{1}{2}\norm{Lu}^2_V$. Then it follows (see e.g.
 \cite[Lem. 2.4]{frick:FriSch10}) that the subdifferential $\partial J(\hat u)$
 consists of the single element $L^*L \hat u$. Hence the extremality relations
 \eqref{impl:kkt} imply that $K^*\hat p = L^*L \hat u$. Now it is easy to
 observe that  
 \begin{equation*}
    J(u_k) - J(\hat u ) - \inner{K^*\hat p}{u_k - \hat u} =
    \frac{1}{2}\norm{L(u_k - \hat u)}^2_V.
 \end{equation*} 
\end{proof}

\bibliographystyle{abbrv}
\bibliography{literature}
     
\end{document}

%% file: makros.tex
\usepackage{bbm}
\usepackage{amssymb,amsthm,amsmath}
\usepackage{graphicx, enumerate}
\usepackage{psfrag}
\usepackage{array}
\usepackage[square, comma, numbers]{natbib}
\usepackage[unicode, a4paper]{hyperref}
\usepackage{hypernat}
\usepackage{subfigure}
\usepackage{algorithm, algorithmic}
\usepackage{url}
\usepackage{mathcomp}
 
\newcommand{\R}{\mathbb{R}}

\newcommand{\N}{\mathbb{N}}
\newcommand{\Z}{\mathbb{Z}}

\newcommand{\W}{\mathcal{W}}
\renewcommand{\S}{\mathcal{S}}

\newcommand{\Ls}[1]{\text{L}^{#1}}
\renewcommand{\L}[1]{\text{L}^{#1}(\Omega)}

\newcommand{\D}{\text{D}}
\newcommand{\id}{\textnormal{Id}}

\DeclareMathOperator*{\argmin}{\textnormal{argmin}}

\newcommand{\E}[1]{\mathbf{E}\left(#1\right)}

\newcommand{\Prob}{\mathbb{P}}

\newcommand{\eps}{\varepsilon}
\newcommand{\sign}{\text{sign}}
\renewcommand{\ker}{\textnormal{ker}}
\newcommand{\ran}{\textnormal{ran}}

\newcommand{\set}[1]{\left\{ #1 \right\}}
\newcommand{\abs}[1]{\left| #1 \right|}
\newcommand{\norm}[1]{\left\| #1 \right\|}
\newcommand{\inner}[2]{\left\langle #1, #2 \right\rangle}

\newcommand{\smallo}{\text{o}}
\newcommand{\bigo}{\mathcal{O}}
\newcommand{\ra}{\rightarrow}

 \def\vec#1{\ensuremath{\mathchoice
                     {\mbox{\boldmath$\displaystyle\mathbf{#1}$}}
                     {\mbox{\boldmath$\textstyle\mathbf{#1}$}}
                     {\mbox{\boldmath$\scriptstyle\mathbf{#1}$}}
                     {\mbox{\boldmath$\scriptscriptstyle\mathbf{#1}$}}}}%

\def\vec#1{\ensuremath{\mathchoice
                     {\mbox{\boldmath$\displaystyle#1$}}
                     {\mbox{\boldmath$\textstyle#1$}}
                     {\mbox{\boldmath$\scriptstyle#1$}}
                     {\mbox{\boldmath$\scriptscriptstyle#1$}}}}%

\theoremstyle{plain}

\newtheorem{thm}{Theorem}[section]
\newtheorem{prop}[thm]{Proposition}

\newtheorem{assa}{Assumption}

\theoremstyle{definition}
\newtheorem{rem}[thm]{Remark}

\newtheorem{lem}[thm]{Lemma}

\newtheorem{cor}[thm]{Corollary}
\newtheorem{example}[thm]{Example}
\newtheorem*{example*}{Example}

\newtheorem*{dfn*}{Definition}
\newtheorem*{alg*}{Algorithm}

\theoremstyle{remark}